%% file: v2.tex
\documentclass[11pt, oneside]{article}   	
\usepackage{geometry}                		
\usepackage{graphicx}	
\newcommand\ignore[1]{}			
\usepackage{amssymb}
\usepackage{amsmath}
\usepackage{amssymb}
\usepackage{amsthm}
\usepackage{graphicx}
\usepackage[urlcolor=blue,colorlinks=true]{hyperref}
\usepackage{placeins}
\usepackage{tocloft}
\usepackage{slashed}
\usepackage{float}
\usepackage{fullpage}
\usepackage{braket}
\usepackage{url}

\usepackage{tabularx}

\usepackage{cite}
\usepackage{subcaption}
\usepackage{upgreek}
\usepackage{outlines}
\usepackage{braket}
\usepackage{appendix}
\def\0{{(0)}}
\def\1{{(1)}}
\usepackage{xspace}

\newcommand{\ran}[1]{\textcolor{purple}{ \sc *** #1 ***}}

\def\bbP{\mathbb{P}}

\def\d{\partial}

\def\mn{{\mu\nu}}

\def\a{\alpha}
\def\b{\beta}

\usepackage{bm}

\def\m{\mu}

\def\G{\Gamma}
\def\g{\gamma}
\def\k{\kappa}
\usepackage{caption}
\usepackage{subcaption}

\def\im{\operatorname{Im}}

\def\x{\chi}

\def\t{\tau}

\def\Om{\Omega}

\def\onec{\left(\begin{array}{c}}
\def\cend{\end{array}\right)}
\def\cc{\left(\begin{array}{cc}}
\def\ccend{\end{array}\right)}
\def\ccc{\left(\begin{array}{ccc}}

\def\cccend{\end{array}\right)}
\def\cccc{\left(\begin{array}{cccc}}
\def\cccc{\left(\begin{array}{ccccc}}

\def\ccccend{\end{array}\right)}

\def\t{\tau}

\newcommand{\calK}{\mathcal{K}}

\renewcommand\mod[1]{\text{ mod }#1}

\usepackage{tikz}
\usepackage{pgfplots}
\usetikzlibrary{matrix,arrows,decorations.pathmorphing}
\usetikzlibrary{patterns}
\usetikzlibrary{shapes.misc}
\usetikzlibrary{trees}
\usetikzlibrary{decorations.pathmorphing}
\usetikzlibrary{shapes.geometric}
\usetikzlibrary{positioning}
\usetikzlibrary{decorations.markings}
\usetikzlibrary{positioning,arrows}
\newcommand\numberthis{\addtocounter{equation}{1}\tag{\theequation}}
\usetikzlibrary{decorations.pathreplacing}
\usetikzlibrary{decorations.pathmorphing}
\usetikzlibrary{shapes}
\usetikzlibrary{arrows.meta}
\usetikzlibrary{plotmarks}
\usetikzlibrary{decorations.markings}
\tikzset{
particle/.style={thin,draw=black, postaction={decorate},
decoration={markings,mark=at position .5 with {\arrow[black, line width=0.5mm]{stealth}}}},
gluon/.style={decorate, draw=black, decoration={coil,amplitude=4pt, segment length=5pt}},
photon/.style={decorate, decoration={snake}},
singularity/.style={decorate, draw=black, decoration=zigzag}
}

\newcommand{\eq}[1]{\begin{align}#1\end{align}}
\newcommand{\seq}[1]{\begin{align*}#1\end{align*}}
\newcommand{\subeqs}[1]{\begin{subequations}\begin{align}#1\end{align}\end{subequations}}

\def\Z{\ensuremath{\mathbb{Z}}}

\usepackage[bottom]{footmisc}

\newcommand{\SL}{\operatorname{SL}}

\newcommand{\sltz}{\SL(2,\Z)}

\theoremstyle{definition}
\newcommand{\smallmat}[1]{\left(\begin{smallmatrix}#1\end{smallmatrix}\right)}

\def\flux{{\text{flux}}}

\renewcommand{\Z}{\mathbb{Z}}

\usepackage{amsthm}
\theoremstyle{theorem}

\usepackage{authblk}

\title{Supersymmetric Flux Compactifications and Calabi-Yau Modularity}
\author{Shamit Kachru, Richard Nally, and Wenzhe Yang}
\affil{Stanford Institute for Theoretical Physics,\\ Stanford University, Stanford, CA, 94305}

\def\bbZ{\mathbb{Z}}
\def\bbF{\mathbb{F}}
\def\bbP{\mathbb{P}}
\def\bbR{\mathbb{R}}
\def\bbC{\mathbb{C}}
\def\bbQ{\mathbb{Q}}

\usepackage{graphicx}
\def\bR{\reflectbox{R}}
\def\bD{\reflectbox{D}}
\usepackage{hyperref}
\hypersetup{colorlinks,linkcolor={blue},citecolor={blue},urlcolor={red}}
\def\gal{\operatorname{Gal}}
\def\qB{\bar{\bbQ}}
\def\het{H_{\text{\'et}}}
\numberwithin{equation}{section}
\def\fr{\operatorname{Fr}}
\def\gal{\operatorname{Gal}}
\def\x{\chi}
\def\bbQb{\bar{\bbQ}}
\def\calK{\mathcal{K}}

\newtheorem{example}{\text{Example}}[section]
\newtheorem{theorem}{Theorem}[section]
\newtheorem{conjecture}[theorem]{Conjecture}
\newtheorem{remark}[theorem]{Remark}
\newtheorem{definition}[theorem]{Definition}

\newtheorem{corollary}[theorem]{Corollary}
\usepackage{tikz-cd}
\usetikzlibrary{cd}
\usepackage{graphicx, float, array, xspace, amscd, amsmath, hyperref, amsfonts, amsthm, amssymb, latexsym, bbm, booktabs,mathrsfs,indentfirst,dsfont,slashed}

\begin{document}
\maketitle

\begin{abstract}
Flux compactification of IIB string theory associates special points in Calabi-Yau moduli space to choices of (pairs of) integral three-form fluxes.  In this paper, we propose that supersymmetric flux vacua are modular.  That is, to a supersymmetric flux vacuum arising in a variety defined over $\mathbb{Q}$, we associate a two-dimensional Galois representation that we conjecture to be modular. We provide numerical evidence for our conjecture by examining flux vacua arising on the octic hypersurface in $\mathbb{P}^{4}(1,1,2,2,2)$.


\end{abstract}

\clearpage
\tableofcontents

\section{Introduction}
The Langlands program, proposed by R. Langlands in the sixties, is a series of far-reaching conjectures about the mysterious connections between number theory, geometry and analysis \cite{Langlands}. It relates absolute Galois groups in algebraic number theory to automorphic forms and representations of algebraic groups over local fields and adeles. The last 50 years have seen numerous breakthroughs in this area \cite{langlandsBook}. One of the most important achievements in this direction is the proof of the modularity of elliptic curves \cite{Diamond,Conrad,Taylor2,Wiles,Breuil}, which implies that all elliptic curves over $\bbQ$ are associated to a wieght-two newform. Elliptic curves are Calabi-Yau (CY) one-folds, so it is natural to ask whether these results can be generalized to higher-dimensional Calabi-Yau varieties.

At its core, modularity is about the Galois representations associated to the \'etale cohomologies of algebraic varieties. The easiest way to establish modularity for higher dimensional varieties is to find varieties whose associated representations contain a subrepresentation ``similar" (in a technical sense) to those associated to elliptic curves, and then apply existing results to the subrepresentation to prove modularity. This approach has been applied to the two-dimensional case of K3 surfaces to prove the modularity of singular K3s.  More precisely, the transcendental cycles of a singular K3 surface generate a two-dimensional Galois representation that is modular, and associated to it is a wieght-three newform \cite{Schutt,livne}. In dimension three, the first known result is the modularity of rigid CY threefolds; as proven in \cite{yui:rigid}, rigid threefolds are associated to weight-four modular forms.

CY threefolds play a central role in physics, where they feature prominently in string theory. In this paper, we will relate string theory to the modularity of some nonrigid threefolds.  Flux compactification of IIB string theory picks out special points in the complex structure moduli space ${\cal M}(X)$ of a CY threefold $X$ (for reviews, see for instance \cite{rev1,rev2,rev3,rev4}). These points $\phi$ are determined by the choice of two elements $f$ and $h$ of the middle integral cohomology of $X$. In a subset of flux compactifications known as supersymmetric flux compactifications, the complex structure of $X_\phi$ aligns such that $f$ and $h$ span a two-dimensional subspace of the middle singular cohomology which is purely of Hodge type (2,1) + (1,2)\footnote{Though this statement may seem unfamiliar, we prove shortly that it is implied by the usual criteria specifying a supersymmetric flux vacuum.}. This split, combined with the Hodge conjecture, implies the existence of a two dimensional Galois sub-representation, which, if defined over $\bbQ$, is known to be modular and associated to a weight-two cuspidal Hecke eigenform $f_2$ for some congruence subgroup $\Gamma_0(N)$, where the level $N$ should be determined by the primes of bad reduction of $X_\phi$. Thus, threefolds $X_\phi$ defined over $\bbQ$ and admitting a supersymmetric flux compactification are excellent candidates for weight-two modularity.\ignore{is potentially modular!} \footnote{An argument which is similar in spirit -- but physically and mathematically distinct for general $X$ -- was recently used by Candelas, de la Ossa, Elmi, and van Straten to study the modularity of rank-two attractor points (singled out by the ``attractor mechanism" of black hole physics) in a one-parameter family of Calabi-Yau threefolds \cite{candelas:attractors}.}

The Hodge conjecture remains unproven, so we will support this claim with a rich example. In particular, we will carefully consider the octic CY threefold in $\bbP^{4}(1,1,2,2,2)$. This manifold admits a one-dimensional family of supersymmetric flux vacua \cite{dewolfe:11222}. By the argument above, we expect that some rational points in this family should be associated to a weight-two Hecke eigenform. Using the calculation of the $\zeta$-function of this family of threefolds in \cite{kadir:octic}, we have verified modularity for a number of rational points. To the best of our knowledge, the modularity of these points in moduli space was not previously known, but in principle these relationships are already present in \cite{kadir:octic}. Without the analysis in the previous paragraph, however, there is no reason to suspect that rational points along the supersymmetric locus would be modular in this way. We therefore expect that our analysis will help in developing the study of potential connections between physics and the modularity of Calabi-Yau manifolds.

 The outline of this paper is as follows. We begin by reviewing some basic aspects of modularity in Section \ref{sec:modularity}. We then introduce flux compactifications, including supersymmetric vacua, in Section \ref{sec:flux}. The facts in these two sections will enable us to state our main idea, which we do in Section \ref{sec:main}. We support this argument with the example of the octic in Section \ref{sec:octic}, before concluding in Section \ref{sec:conclusion} with discussion and outlook. We also provide several appendices. In Appendix \ref{sec:mirror}, we review the relationship between the $\zeta$-function of a manifold and that of its mirror, which will be used throughout the paper. In Appendix \ref{sec:wenzhe}, we provide a brief overview of the arithmetic and algebraic geometry necessary to make the arguments of Section \ref{sec:main} rigorous. Finally, in Appendix \ref{sec:data} we provide tables of data supporting the modularity claims of Section \ref{sec:octic}.

\section{The Modularity of Calabi-Yau Varieties}
\label{sec:modularity}
In this section we will briefly review the relationship between algebraic varieties and automorphic forms, known simply as modularity. We will focus on elliptic curves and Calabi-Yau threefolds defined over $\mathbb{Q}$; readers are referred to \cite{langlandsBook} for a much more general introduction to the Langlands program. This is the basic setting for our work, but may be unfamiliar to physicists, so we will review it in some detail; readers  familiar with this material can skip this section completely. We will take a fairly concrete perspective, and focus on counting points on projective varieties; a much more abstract perspective is provided in Appendix \ref{sec:wenzhe}. First, in Section \ref{sec:modEll}, we will review the famous modularity theorem for elliptic curves, following \cite{ribet:review}. Next, in Section \ref{sec:modThree}, we will discuss recent progress in the study of modularity for threefolds, following \cite{meyer:book,yui:review}. Finally, in Section \ref{sec:WeilConj} we will provide an alternative perspective on these results that will be essential later.

Before we proceed, let us review some basic facts about modular forms; for a more thorough introduction see e.g. \cite{123}. A weight-$w$ modular form  for a discrete subgroup $\Gamma\subset\operatorname{SL}(2,\bbR)$ is a holomorphic function $f:\mathbb{H}\to\bbC$ satisfying the functional equation \eq{f\left(\frac{a\t+b}{c\t+d}\right) = \left(c\t+d\right)^wf(\t)} for all $\smallmat{a&b\\c&d}\in\Gamma$. For our purposes, we will always take $\G$ to be of the form $\G_0(N)$ for some positive integer $N$, i.e. the matrix group \eq{\G_0(N) := \left\{\left.\left(\begin{array}{cc} a&b\\c&d\end{array}\right)\in\operatorname{SL}(2,\bbZ) \right| c\equiv 0\text{ mod }N\right\}.} The matrix $T\equiv\smallmat{1&1\\0&1}$, which is a generator of $\sltz$, is in $\G_0(N)$ for all $N$, so all modular forms for $\G_0(N)$ have a Fourier expansion of the form \eq{f(\t) = \sum_{n=0}^\infty c_nq^n,} where $q=\exp\left(2\pi i\t\right)$. We say that $f$ is a cusp form if it vanishes at the ``cusps" $\{i\infty\cup\bbQ\}$; the condition that $f$ vanishes at infinity implies that all cusp forms have $c_0=0$. To any cusp form $f$, we can associate an $L$-function $L(f,s)$, defined by \eq{L(f,s) = \sum_{n=1}^\infty c_nn^{-s}.\label{eq:L(f)}}

We write $S_w(N)$ for the vector space of weight-$w$ cusp forms for $\G_0(N)$. The spaces $S_w(N)$ are acted on by endomorphisms $T_n$ for $n\ge1$ known as ``Hecke operators." We say a cusp form $f$ is a Hecke eigenform if it is an eigenvector under all Hecke operators \cite{123}. Having established these basic definitions, we can now describe the connections between counting points on projective varieties and Hecke eigenforms; this will be the content of the remainder of this section.

\subsection{Modularity For Elliptic Curves}
\label{sec:modEll} 

Physicists are most familiar with elliptic curves as complex manifolds, i.e. tori. However, here we will focus on the subset of elliptic curves that can be defined by polynomial equations with rational coefficients; we will say such elliptic curves are defined over $\bbQ$. Any elliptic curve $E$ over $\bbQ$ can be represented in the Weierstrass form as \eq{y^2 + \a_1xy + \a_3y = x^3 + \a_2x^2 + \a_4x+\a_5,\label{eq:weierstrass}} where all of the $\a_i$ are integers. \ignore{It is more convenient to repackage this as a projective polynomial by introducing an auxiliary variable $z$: \eq{y^2z+ \a_1xyz + \a_3yz^2 = x^3 + \a_2x^2z + \a_4xz^2+\a_5z^3\label{eq:weierstrassProj}.}} $E$ is smooth so long as there does not exist a point $(x,y)\in{X}$ at which all partial derivatives of Eq. \ref{eq:weierstrass} vanish.

By simply reducing Eq. \ref{eq:weierstrass} modulo a prime $p$, we obtain a curve over the finite field $\bbF_p$, called $E/\bbF_p$. Even if $E$ is smooth, $E/\bbF_p$ might be singular over $\bbF_p$. There will only be finitely many primes at which this happens; these are called the primes of bad reduction of $E$, or simply bad primes, and the primes at which $E/\bbF_p$ is nonsingular are called good primes. For instance, the elliptic curve \eq{y^2 = x^3 + x^2 -77x -289\label{eq:ellCurve}} has bad reduction at $p=2$ and $p=11$ \cite{interactiveTable}. From the bad primes of $E$, we can define an important arithmetic invariant $N$, called the conductor of $E$.

Over the field $\bbF_p$, Eq. \ref{eq:ellCurve} only has finitely many solutions; we will denote this number as $\#(E,p)$. Importantly, we include the $x=y=\infty$ solution of \ref{eq:ellCurve}. It is convenient to repackage these numbers as \eq{a_p \equiv p + 1 - \#\left(E,p\right).\label{eq:apEll}} 
The modularity theorem \cite{Wiles,Taylor2,Diamond,Conrad,Breuil} states that, for all elliptic curves $X$ over $\bbQ$, there exists a weight-two cuspidal Hecke eigenform \eq{f_E=\sum_{n=1}^\infty c_n q^n} for $\Gamma_0(N)$, where $N$ is the conductor of $E$, such that \eq{a_p=c_p\label{eq:apcpEll}}
for all good primes $p$. In the example of Eq. \ref{eq:ellCurve}, the numbers $a_p$ agree with the Fourier coefficients $c_p$ of the weight-two Hecke eigenform for $\Gamma_0(44)$ with Fourier expansion \cite{interactiveTable} \eq{f(\t) = q + q^{3} - 3q^{5} + 2q^{7} - 2q^{9} +\cdots.} 

The modularity theorem can be stated more succinctly in terms of $L$-functions. For good primes $p$ of $E$, we define a local $L$-factor \eq{L_p(E,t)=\left( 1-a_p t+pt^2 \right)^{-1}. \label{eq:localEll}} We can similarly define a local $L$-factor at a bad prime $p$. In terms of the local $L$-factors, the $L$-function of $E$ is defined as the infinite product \eq{L(E,s)=\prod_{p ~\text{bad}} L_p(E,p^{-s}) \cdot \prod_{p ~\text{good}} L_p (E,p^{-s}). \label{eq:eulerEll}} A priori, this infinite product only converges on the right half plane with sufficiently large $\operatorname{Re}(s)$. Readers are referred to the Appendix \ref{sec:wenzhe} for more details. The modularity theorem can then be written as \eq{L(E,s)=L(f_E,s),\label{eq:modEll}} where $L(f_E,s)$ was defined in Eq. \ref{eq:L(f)}.

\ignore{defined a coefficient $a_p$ in Eq. \ref{eq:apEll}; this definition can be extended to both bad primes $p$ and composite integers $n$, so we can define an $L$-function $L(E,s)$ as \eq{L(E,s) = \sum a_nn^{-s}.} A priori, this sum only converges for sufficiently large $\operatorname{Re}(s)$.

It will be useful to rephrase the definition of $L$-functions as an infinite product. For a good prime $p$ of $E$, we define the local $L$-factor \eq{L_p(E,t)=1-a_p t+pt^2. \label{eq:localEll}} We can similarly define a local $L$-factor at a bad prime $p$. In terms of the local $L$-factors, the $L$-function of $E$ is by definition given by the infinite product \eq{L(E,s)=\prod_{p ~\text{bad}} L_p(E,p^{-s})^{-1} \cdot \prod_{p ~\text{good}} L_p (E,p^{-s})^{-1}. \label{eq:eulerEll}}}

We conclude by briefly mentioning a more abstract perspective on this material. The middle \'etale cohomology $H^1_{\text{\'et}}(E)$ furnishes a two-dimensional representation $\rho_2(E)$ of the absolute Galois group $\gal(\qB/\bbQ)$, i.e. the automorphism group of the maximal algebraic extension of $\bbQ$. Frequently, the best way to study modularity is to study abstract properties of this representation. Indeed, the modular form $f_E(q)$ is itself best thought of as being associated to another representation, which corresponds to $\rho_2(E)$ under the Langlands correspondence. However, we will try to avoid this language, and only invoke it when absolutely necessary.

\subsection{Modularity For Calabi-Yau Threefolds}
\label{sec:modThree}
Elliptic curves are CY onefolds, so it is natural to ask to what extent these modularity results can be generalized to higher-dimensional Calabi-Yau varieties. For K3 surfaces, one analog of the modularity theorem of elliptic curves is known: all singular K3 surfaces are modular \cite{livne}. In singular K3s, the two-dimensional transcendental lattice induces a two-dimensional subspace of the middle \'etale cohomology which is purely of Hodge type $(2,0)+(0,2)$. This causes the Galois representation associated to the middle \'etale cohomology to split into the sum of a two-dimensional subrepresentation and a twenty-dimensional one; the two-dimensional summand is modular, and associated to a weight-three Hecke eigenform.

For threefolds, the situation is much more complicated. The middle \'etale cohomology of a Calabi-Yau threefold $ X$ defined over $\bbQ$, denoted by $\het^3(X)$, furnishes a $b_3=2+2h^{2,1}$-dimensional representation $\rho_{b_3}(X)$ of $\gal(\qB/\bbQ)$. While two-dimensional Galois representations are comparatively well-understood, higher-dimensional representations are less well-studied, and so for threefolds with large $b_3$ these representations are too complicated to study directly. On the other hand, if $X$ is rigid, i.e. it has $h^{2,1}=0$, then $\het^3(X)$ is two-dimensional, and we can study it directly. Indeed, it was proven by Gouvea and Yui that all rigid threefolds over $\bbQ$ are modular \cite{yui:rigid}.

To state this result more precisely, we must define an $L$-function for rigid threefolds. As we did with elliptic curves above, for a good prime $p$ of a rigid threefold $X$ we have a local $L$-factor  \eq{L_p(X,t)=\left( 1-a_pt+p^3t^2 \right)^{-1},\label{eq:localThree}} where \cite{meyer:book} \eq{a_p = p^3 + \left(p^2+p\right)k_p(X) - \#(X,p),\label{eq:ap3}} with $|k_p|<h^{1,1}(X)$. We similarly have local $L$-factors for bad primes, so as before we define the $L$-function by \eq{L(X,s) = \prod_{p~\text{bad}}L_p\left(X,p^{-s}\right) \cdot\prod_{p~\text{good}}L_p\left(X,p^{-s}\right).\label{eq:eulerRigid}} Then the modularity of a rigid threefold $X$ just means that there exists a weight-four Hecke eigenform $f_4(\t) = \sum_n b_nq^n$ for some $\Gamma_0(N)$ such that \cite{yui:rigid} \eq{L(X,s) = L(f_4,s)\label{eq:modRigid}.} Here, $N$ is not simply the conductor of $X$, and indeed does not admit a sharp characterization, but as with elliptic curves it is divisible only by the primes of bad reduction of $X$. As before, we can expand  the infinite product in Eq. \ref{eq:eulerRigid} into an infinite sum of the form \eq{L(X,s) = \sum_n a_n n^{-s}.} Eq. \ref{eq:modRigid} then implies that \eq{a_p = b_p \label{eq:apbpRigid}} for all good $p$.

For nonrigid threefolds, the situation is less nice, and is in general not yet tractable. One can make progress, however, by finding three-dimensional analogs of singular K3 surfaces, i.e. algebraic threefolds associated to two-dimensional subrepresentations. More precisely, suppose the $b_3$-dimensional Galois representation $\rho_{b_3}(X)$ splits into the direct sum \eq{\rho_{b_3}(X) = \rho_2 (X) \oplus \rho_{b_3-2}(X)\label{eq:galoisSplit}.} The summand $\rho_2(X)$ is somewhat similar to the Galois representation associated to a rigid threefold, and was studied in \cite{yui:rigid} with similar techniques. More precisely, if $\rho_2(X)$ is of Hodge type $(3,0)+(0,3)$, then the situation is identical to the rigid case, and $\rho_2(X)$ is associated to a weight-four eigenform; many examples of such a split can be found in e.g. \cite{yui:review,meyer:book}. It was pointed out in \cite{candelas:attractors} that such splits can be related to rank-two attractor points in Calabi-Yau moduli space.

If on the other hand $\rho_2(X)$ is of Hodge type $(2,1)+(1,2)$, then it is still modular. To a Galois representation $\rho$, we can associate a related representation, called the Tate twist\footnote{The Tate twist can be thought of as a way to keep track of more sensitive algebraic information about a variety. For example, given an arbitrary smooth Calabi-Yau threefold $X$ defined over $\mathbb{Q}$, we know $H^0(X,\mathbb{Q})=H^6(X,\mathbb{Q})=\mathbb{Q}$. On the other hand, the \'etale cohomologies are given by $H^0_{\text{\'et}}(X)=\mathbb{Q}_\ell(0)$ and $H^6_{\text{\'et}}(X)=\mathbb{Q}_\ell(-3)$, so because of the Tate twists we are able to distinguish between the two cohomology groups.}  $\rho\otimes\bbQ_\ell(1)$ of $\rho$; readers are referred to Appendix \ref{sec:wenzhe} for more details about Tate twists. If $\rho$ is two-dimensional and has Hodge type $(a,b)+(b,a)$, then $\rho\otimes\bbQ_\ell(1)$ has Hodge type $(a-1,b-1) + (b-1,a-1)$. The Tate twist $\rho_2(X) \otimes \mathbb{Q}_\ell(1)$ of $\rho_2(X)$ is therefore of Hodge type (1,0)+(0,1), which is the same Hodge type as the middle \'etale cohomology of an elliptic curve. We then have that $\rho_2(X) \otimes \mathbb{Q}_\ell(1)$ is modular by \cite{Wiles,Taylor2,Diamond,Conrad,Breuil}. In particular, it is associated to a weight-two eigenform \eq{f_2(\t) = \sum_n c_n q^n.} In this paper, we will relate this notion of modularity to supersymmetric flux compactifications.

\ignore{\begin{remark}
The Tate twist can be considered as a way to keep track of the more sensitive algebraic information of a variety. For example, given an arbitrary smooth Calabi-Yau threefold $X$ defined over $\mathbb{Q}$, we know $H^0(X,\mathbb{Q})=H^6(X,\mathbb{Q})=\mathbb{Q}$. But these two cohomology groups carry different weights. For its \'etale cohomology, $H^0_{\text{\'et}}(X)=\mathbb{Q}_\ell(0)$, which has weight $0$; $H^6_{\text{\'et}}(X)=\mathbb{Q}_\ell(-3)$, which has weight $6$. These different weights indeed show up when we discuss the arithmetic properties of $X$.  \ran{Does the notion of weight show up anywhere else in the main text? I don't think so- I'd shorten it to something like "$H^0(\bbQ) = H^6(\bbQ) = \bbQ$. On the other hand, the \'etale cohomologies are given by $H^0_{\text{\'et}}(X)=\mathbb{Q}_\ell(0)$ and $H^6_{\text{\'et}}(X)=\mathbb{Q}_\ell(-3)$. Thus, the Tate  object $\bbQ(n)$ allows us to keep trackof the more sensitive algebraic information of a variety.} 
\end{remark}}

For nonrigid threefolds, we can again define an $L$-function as a  product over local $L$-factors, exactly as in Eq. \ref{eq:eulerRigid}. The $L$-factor $L_p(X,t)$ associated to a good prime $p$ will now be the inverse of a polynomial of degree $b_3(X)$. If $X$ has an \'etale  split of the kind we have considered here, then the local $L$-factors  will factorize over $\bbZ$, and we will have that \eq{L_p(X,t) = L_p^{(2)}(X,t)L_p^{(b_3-2)}(X,t),}  where $L_p^{(2)}$ is the inverse of a quadratic in $t$ and the inverse of $L_p^{(b_3-2)}$ has degree $b_3-2$. The quadratic factor encodes the Fourier  coefficients of the associated modular form. If $\rho_2(X)$ has Hodge type $(3,0) + (0,3)$ then we will have \eq{L^{(2)}_p(X,t) = \left(1-b_pt + p^3t^2\right)^{-1},} where $b_p$ is the $p$-th Fourier coefficient of a weight-four eigenform $f_4(\t) = \sum_n b_nq^n$, and on the other hand if $\rho_2(X)$ has Hodge type $(2,1) + (1,2)$ then we will have \eq{L_p^{(2)} =\left( 1-c_p(pt)+p(pt)^2 \right)^{-1},} where $c_p$ is the $p$-th coefficient of a weight-2 eigenform $f_2(\t) = \sum_n c_nq^n$.

However, at the level  of point counts, what we compute for a good prime $p$ is not the local  $L$-factor  $L_p(X,t)$, but instead the point count coefficient $a_p$ defined in Eq. \ref{eq:ap3}. Thus, it is more natural to expand the $L$-function as a sum, \eq{L(X,s) = \sum  a_n n^{-s},} and look for an explicit relationship between the $a_p$ and the Fourier coefficients of the associated threefolds, along the lines of Eq. \ref{eq:apbpRigid}. For threefolds associated to a weight-four newform, this is easy:  we simply have \cite{meyer:book} \eq{a_p \equiv b_p\ \operatorname{mod}{p}.} On the other hand, if instead $X$  is associated to a weight-two newform, there is no such simple relationship. This is the notion of modularity we are interested in, so it will be inconvenient for us to use $L$-functions in examples. In the next section, we will encounter a more subtle diagnostic of modularity that will be more useful.

\subsection{$\zeta$-Functions and The Weil Conjectures} \label{sec:WeilConj}
So far, we have attempted to provide a very concrete introduction to $L$-functions and modularity. However, going forward, $L$-functions will not be especially helpful. Instead, we will need to introduce the slightly more abstract construction of the $\zeta$-function.

To begin, let us formalize the notion of a variety over a finite field. Suppose $X$ is an $n$-dimensional\footnote{Here, $n$ is the complex dimension of $X$.} non-singular variety defined over the field $\mathbb{Q}$ of rational numbers; intuitively, $X$ can be thought of as the vanishing locus of one or more polynomial equations with rational coefficients. If we multiply these polynomial equations by the least common multiple of the denominators of their coefficients, we will obtain polynomial equations with integral coefficients that define a variety $\mathcal{X}$ over $\mathbb{Z}$, which is called the integral model of $X$. Given a prime number $p$, modulo $p$ the integral model $\mathcal{X}$ defines a variety over the finite field $\mathbb{F}_p:=\mathbb{Z}/p\mathbb{Z}$, which we will denote by $X/\mathbb{F}_p$. We say $p$ is a good prime of $X$ if the variety $X/\mathbb{F}_p$ is non-singular, and a bad prime otherwise. 

Suppose $p$ is a good prime of $X$ and $m$ is a positive integer. Up to isomorphism there is a unique degree-$m$ extension of $\mathbb{F}_p$ that will be denoted by $\mathbb{F}_{p^m}$. Since $\mathbb{F}_p$ is a subfield of $\mathbb{F}_{p^m}$, the variety $X/\mathbb{F}_p$ is naturally also a variety defined over $\mathbb{F}_{p^m}$. Over $\bbF_{p^m}$, $X/\mathbb{F}_{p^m}$ has a finite number of solutions; we denote this number by $\#(X,p^m)$. We now define the $\zeta$-function $\zeta(X,p,t)$ as the generating series in the formal variable $t$ given by
\begin{equation}
\zeta(X,p,t):=\exp \left( \sum_{m=1}^{\infty} \frac{\#\left(X,p^m\right)}{m}  \,t^m     \right). \label{eq:zetaDef}
\end{equation}
Although a priori this is only a formal power series, the $\zeta$-function enjoys many nice analytic properties; these are summarized by the Weil conjectures. The first conjecture, known as the rationality conjecture and proven by Dwork using $p$-adic analysis \cite{Dwork}, holds that the $\zeta$-function is a rational function of $t$ of the form \begin{equation}
\zeta(X,p,t)=\frac{P_1(X,p,t) \cdots P_{2n-1}(X,p,t)}{P_0(X,p,t) \cdots P_{2n}(X,p,t)}, \label{eq:weilRational}
\end{equation}
where each of the $P_i(X,p,t)$ is a polynomial with integral coefficients; we will frequently denote these polynomials  by $P_i(t)$ if there is no ambiguity. The polynomials $P_0(X,p,t)$ and $P_{2n}(X,p,t)$ take particularly simple forms: \begin{equation}
P_0(X,p,t)=1-t,~P_{2n}(X,p,t)=1-p^nt.\label{eq:p0p2n}
\end{equation}
The orders of the polynomials $P_i$ are determined by the Betti numbers of the complex $n$-fold $X(\bbC)$ defined by the variety $X$:
\begin{equation}
\text{deg}\,P_i(X,p,T)=\text{dim}_{\mathbb{Q}}H^i(X(\mathbb{C}),\mathbb{Q}).
\end{equation}
The polynomials $P_i$ factorize over $\bbC$ as \eq{P_i(X,p,t) = \prod_{j}\left(1-\a_{ij}t\right).} The ``Riemann hypothesis" portion of the Weil conjectures implies that \eq{\left|\a_{ij}\right| = p^{i/2},} as was proven by Deligne \cite{DeligneWeil}.

The polynomials $P_i(t)$ can be determined by the \'etale cohomology of $X$. $X$ can be modeled as one or more polynomials in variables $x_i$. For each prime $p$, there is a natural map called the geometric Frobenius action $\operatorname{Fr}_p$, which takes \eq{x_i\to x_i^p.} The Frobenius acts on the \'etale cohomologies of $X$, and the $P_i$ are the characteristic polynomials of this action: \cite{MilneEC}\begin{equation}
P_i(X,p,t)=\det \left(\text{Id}-t\,\text{Fr}_p|_{H^i_{\text{\'et}}(X_{\overline{\mathbb{Q}}},\mathbb{Q}_\ell)} \right).
\end{equation}
Actually, the Frobenius map has already featured prominently in our discussion of $L$-functions: the local $L$-factors in Eqs. \ref{eq:localEll} and \ref{eq:localThree} are the inverse of the characteristic polynomials of the Frobenius action on the middle \'etale cohomology $\het^n(X,\bbQ_\ell)$ of $X$ (with $t$ replaced by $p^{-s}$).

It will be instructive to apply this somewhat abstract discussion to the concrete example of a CY threefold $X$ defined over $\bbQ$ with Hodge diamond \begin{equation} \label{eq:HodgeDiamond}
	\begin{tabular}{ c c c c c c c }
		&  &  & 1 &  &  &  \\ 
		&  & 0&   & 0&  &  \\   
		& 0&  & $h^{1,1}$ &  & 0&   \\  
		1&  & $h^{2,1}$&  & $h^{2,1}$ & & 1 \\ 
		& 0&  & $h^{1,1}$ &  & 0&   \\ 
		&  & 0&   & 0&  &  \\   
		&  &  & 1 &  &  &  \\
	\end{tabular},
\end{equation}
where $h^{1,1}:=\text{dim}_{\mathbb{C}}H^{1,1}(X)$ and $h^{2,1}:=\text{dim}_{\mathbb{C}}H^{2,1}(X)$. The only polynomial that survives in the numerator is $P_3(t)$, of degree $b_3(X) = 2 + 2h^{2,1}$. The denominator has four contributions: the linear factors $P_0(t)$ and $P_6(t)$, whose forms are given in \ref{eq:p0p2n}, and the factors $P_{2}(t)$ and $P_4(t)$, each of degree $h^{1,1}$. Putting it all together, we have that \eq{\zeta(X,p,t) = \frac{R(t)}{D(t)},} where the numerator 

\begin{equation}
R(t) = P_3(t)
\end{equation} 
has degree $2+2h^{2,1}$ and the denominator 
\begin{equation}
D(t) = P_0(t) P_2(t) P_4(t) P_6(t)
\end{equation}
has degree $2+2h^{1,1}$. This form of the $\zeta$-function suggests a connection to mirror symmetry; this connection was explored in \cite{candelas:finite2}, and is summarized in Appendix \ref{sec:mirror}. The upshot is that we can sometimes focus on a factor $R_0(t)$ of $R(t)$ coming from the mirror $Y$ of $X$, and thus limit ourselves to studying a lower-order polynomial. 

The numerator $R(t)$ is the characteristic polynomial of the Frobenius on the middle \'etale cohomology $\het^3(X,\bbQ_\ell)$ of $X$. If this cohomology splits, then $R(t)$ will factor over $\bbZ$. More precisely, if, as in Eq. \ref{eq:galoisSplit}, the middle \'etale cohomology contains a two-dimensional summand $\rho_2$, then $R(t)$ will contain a quadratic factor, whose form depends on the Hodge type of $\rho_2$. In particular, if $\rho_2$ has Hodge type $(3,0)+(0,3)$, then $R(t)$ will contain a factor of the form \eq{ 1 - b_p t + p^3t^2 \Big| R(t),} where $b_p$ is the $p$-th Fourier coefficient of a weight-four Hecke eigenform $f_4(\t) = \sum b_n q^n$ associated to $\rho_2$. On the other hand, if $\rho_2$ has Hodge type $(2,1)+(1,2)$, then the factorization will be of the form \eq{ 1 - c_p(pt) + p(pt)^2 \Big| R(t) ,\label{eq:wt2factor}} where $c_p$ is the $p$-th Fourier coefficient of the weight-two Hecke eigenform $f_2(\t) = \sum c_nq^n$ associated to $\rho_2$. This is the criterion which we were looking for in the previous section. We will use a factorization of this form to diagnose weight-two modularity in examples.

\section{Supersymmetric Flux Compactifications}
\label{sec:flux}
In the previous section, we saw that threefolds $X$ whose middle \'etale cohomology splits are modular. However, identifying points in moduli space at which such a split occurs is difficult. The main point of this paper is that supersymmetric flux compactifications provide examples of such points. To explain why, we will need to recall some details about flux compactifications; we will do so following \cite{kachru:gkp,kachru:0411} (which built on the earlier work of \cite{becker:m8fold,dasgupta:mGflux}).

Our basic physical setting is type IIB string theory. At low energies, string theory is well-described by supergravity, a classical theory in which we approximate extended stringy objects as point particles; we will work in the supergravity approximation for the duration of this paper. In the ten dimensional supergravity associated to type IIB string theory, inventively named type IIB supergravity, the basic degrees of freedom are the metric tensor $g_{MN}$, $p$-form gauge potentials $C_0$, $C_2$, and $C_4$, another two-form gauge potential $B_2$, and a scalar called the dilaton $\phi$, as well as several fermionic degrees of freedom required by supersymmetry. To a $p$-form gauge field, we associate a gauge-invariant $(p+1)$-form field strength; we write the field strength of $C_p$ as $F_{p+1}$, and of $B_2$ as $H_3$.

The $p$-form gauge fields are sourced by extended objects that stretch in $p-1$ spatial dimensions in the same way that the one-form gauge field familiar from electromagnetism is sourced by point particles. These objects are called D$(p-1)$ branes. IIB string theory has even-form gauge potentials, and so has D-(odd) branes, e.g. D1, D3, and D5 branes. Of these, D3 branes will be the most relevant. There are also extended objects with negative tension, known as orientifold planes; these will also be extremely relevant to the construction of flux compactifications.

In what follows, it will be convenient to repackage the axion and dilaton into the axiodilaton $\tau$, defined by \eq{\t = C_0+ie^{-\phi},} and then to work with a complexified three-form flux $G_3$, defined by \eq{G_3 = F_3 - \tau H_3 \label{eq:g3def}} We also define \eq{\tilde{F}_5 = F_5 - \frac{1}{2}C_2\wedge H_3 + \frac{1}{2}B_2\wedge F_3.} We force the field strength $F_5$ to obey a self-duality constraint \eq{\tilde{F}_5 = \star \tilde{F}_5\label{eq:selfdual}} by hand. The dynamics of these fields, excluding the self-duality constraint, are summarized in the effective action \seq{S =\ &\frac{1}{2\k_{10}^2}\int d^{10}x\sqrt{-g}\left[R - \frac{\d_M\t\d^M\bar{\t}}{2(\im{\t})^2} - \frac{G_3\cdot\bar{G}_3}{12\im{\t}} - \frac{1}{480}\tilde{F}_5^2\right] +\ignore{\\ +\ &}\frac{1}{8i\k_{10}^2}\int \frac{C_4\wedge G_3\wedge\bar{G}_3}{\im{\t}} + S_{\text{loc}},\numberthis\label{eq:action}} where $S_{\text{loc}}$ contains the action of localized objects such as branes. This action makes the $\sltz$ symmetry \eq{\t\to\frac{a\t+b}{c\t+d}, \ \ \ G_3 \to\frac{G_3}{c\t+d}} manifest.


The IIB supergravity equations follow by varying the action $S$ and supplementing by the self-duality condition.  
In addition, there is a Bianchi identity for $\tilde F_5$:
\eq{d\tilde{F}_5 = H_3\wedge F_3 + 2\k_{10}^2T_3\rho_3^{\text{loc}},\label{eq:bianchiDiff}} 
where $T_3$ is the tension of a D3 brane and $\rho_3^{\text{loc}}$ is the density of D3 charge provided by localized sources.

\ignore{We find the equations of motion for the ten-dimensional system by varying Eq. \ref{eq:action} with respect to the fields. Most important for our purposes with be the Einstein equations for the ten-dimensional metric, given by\footnote{Following \cite{kachru:gkp}, we give the trace-reversed Einstein equations.} \eq{R_{MN} = \k_{10}^2\left(T_{MN} - \frac{1}{8}g_{MN}T\right),\label{eq:EFE}} where the stress tensor $T^{MN}$ is the sum of both a bulk piece and a piece $T^{MN}_{\text{loc}}$ capturing the contributions from localized objects. }

To obtain a four-dimensional theory from the ten-dimensional one, we take the ten dimensional spacetime $M_{10}$ to be a topological product \eq{M_{10} = \bbR^{3,1}\times X,} where $X$ is a CY threefold, i.e. a compact, three complex dimensional K\"ahler manifold with vanishing first Chern class. We take the Hodge diamond of $X$ to be of the form given in Eq. \ref{eq:HodgeDiamond}. Integrating \eqref{eq:bianchiDiff} over $X$, we have \eq{\int_{X}H_3\wedge F_3 = -2\k_{10}^2T_3Q_3,\label{eq:bianchi}} where $Q_3$ is the total D3 brane charge, which can be sourced by e.g. O3 planes and D7 branes, as well as D3 branes.  It is shown in \cite{kachru:gkp} that a conformally Calabi-Yau
ansatz for the metric on $X$, together with suitable imaginary self-duality conditions on $G_3$, yields supersymmetric flux vacua.

As discussed in \cite{kachru:gkp}, these supersymmetric flux compactifications require that objects carrying negative D3-charge (and also, effectively negative tension) should be present in the compact dimensions.  This is because one can prove that in non-trivial solutions,
one has
\begin{equation}
\int_X H_3 \wedge F_3 > 0~.
\end{equation}
The need for negative charge and tension is fine, as string theory contains suitable objects.
The sources which usually contribute the requisite negative D3 charge include D7-branes wrapping divisors in (an orientifold of) the compact manifold $X$,
or orientifold O3-planes at points in $X$.  The D7 solutions are, without loss of generality, related to F-theory compactified on a CY fourfold $\tilde{X}$ \cite{sen:fourfold}. Such a configuration gives an effective D3 brane charge \eq{Q_{3}^{\text{eff}} = -\frac{\chi\big(\tilde{X}\big)}{24},} where $\x(\tilde{X})$ is the Euler character of $\tilde{X}$. Thus, we can think of Eq. \ref{eq:bianchi} as a tadpole cancellation condition: \eq{\frac{1}{2\k_{10}^2T_3}\int_{X} H_3\wedge F_3 + Q_3^{D3} = \frac{\chi\big(\tilde{X}\big)}{24},\label{eq:tadpole}} where $Q_3^{D3}$ is the contribution from mobile D3 branes, if present. 
 
We are thus led to consider compactifications with flux: in the absence of mobile D3 branes, or even in the presence of insufficiently many mobile D3 branes to fully cancel the tadpole, both $H_3$ and $F_3$ must have nonvanishing flux over $X$ to satisfy Eq. \ref{eq:tadpole}. On the other hand, we cannot have too much flux: in the absence of anti-D3 branes (which cannot be present in the tree-level ``no-scale supergravity" solutions of \cite{kachru:gkp}, though they can play an interesting role in more general solutions), Eq. \ref{eq:tadpole} also bounds the magnitude of the fluxes.  Choosing an integral symplectic basis $\alpha^a,\beta_b$ for the middle cohomology of $X$, we can define integral flux vectors $f$ and $h$ by writing \subeqs{F_3 &= -(2\pi)^2\a'\left(f_a\a^a + f_{a+h_{2,1}+1}\b_a\right) \\ h_3 &= -(2\pi)^2\a'\left(h_a\a^a + h_{a+h_{2,1}+1}\b_a\right).} We will set $(2\pi)^2\a'=1$ for the remainder of the paper. In terms of the symplectic matrix $\Sigma = \smallmat{0&1\\-1&0},$ the bound on the fluxes can be rewritten as \eq{f\cdot\Sigma\cdot{h}\le\frac{\chi\big(\tilde{X}\big)}{24}.}

Once we have specified flux vectors, the dynamics of the moduli fields of interest here are determined by a Kahler potential $\calK$ and a superpotential $W$, first derived in \cite{gukov:superpotential}. These are given in terms of the complexified flux $G_3$ (which is itself determined by the choice of $f$ and $h$), the axiodilation $\tau$, and the holomorphic threeform $\Omega$ as\footnote{Here, we have neglected the Kahler potential for the Kahler moduli.} \subeqs{\calK &= -\ln\left[-i\left(\t-\bar{\t}\right)\right] - \ln\left(-i\int_{X}\Om\wedge\bar{\Om}\right) \label{eq:kahler}\\W &= \int_{X} G_3\wedge\Om\label{eq:W}.} We say a point $\phi$ in complex structure moduli space defines a flux vacuum for the fluxes $f$ and $h$ if the conditions \eq{D_\t W = D_aW = 0\label{eq:flatness}} are met, where the index $a$ runs over the $h^{2,1}$ complex structure moduli of $X$ and \eq{D_IW = \d_IW + W\d_I\calK.\label{eq:kahlerCovariant}} We have used a new index $I$ in Eq. \ref{eq:kahlerCovariant} to emphasize that the same form holds for $D_\t{W}$. One can expand these conditions to find a criterion on the complex structure. In particular, a point $X_\phi$ in moduli space admits a flux compactification iff its complex structure has aligned such that \eq{G_3 \in H^{2,1}(X_\phi) \oplus H^{0,3}(X_\phi).}

We say that a flux vacuum is supersymmetric if, in addition to meeting Eq. \ref{eq:flatness}, it also satisfies\footnote{This condition comes from satisfying the F-term conditions for Kahler moduli; see \cite{kachru:gkp}.} \eq{W = 0.} This implies that $G_3$ has no (0,3) part, so we simply have \eq{G_3 \in H^{2,1}(X_\phi).\label{eq:g3h21}} Eq. \ref{eq:g3h21} is a very stringent criterion on the complex structure of $X_\phi$, and is the starting point for our analysis. However, it will be convenient to repackage this criterion somewhat. 

For a supersymmetric flux compactification, the constraint $D_\t{W}=0$ becomes \eq{0 = D_\t W = \d_\t W + W\d_\t\calK = -\int_{X_\phi}h\wedge\Om + \frac{-i}{\t-\bar{\t}}W = -\int_{X_\phi}h\wedge\Om \label{eq:DtauWsusy},} so in a $W=0$ compactification the flux vector $h$ has no (0,3) part. Taking the complex conjugate of Eq. \ref{eq:DtauWsusy}, we have \eq{\int_{X_\phi}h\wedge\bar{\Om}=0,} i.e. $h$ has no (3,0) part. Thus \begin{subequations} \eq{h \in H^{2,1}(X_\phi)\oplus H^{1,2}(X_\phi).} On the other hand, in light of Eq. \ref{eq:g3def}, if $f$ had any (3,0) or (0,3) part, then so would $G_3$, so we must also have \eq{f \in H^{2,1}(X_\phi)\oplus H^{1,2}(X_\phi).} \label{eq:fh21plus12}\end{subequations}

Thus, the complex structure alignment necessary to fix the Hodge type of $G_3$ also restricts the Hodge type of $f$ and $h$;  this consequence of Eq. \ref{eq:g3h21} has been observed before, in e.g. \cite{dewolfe:11222}. In particular, arbitrary complex superpositions of these two flux vectors will remain in $H^{2,1}(X_\phi)\oplus H^{1,2}(X_\phi)$. 
We note in passing that this is somewhat analogous to the attractor equation for a supersymmetric black hole in compactification on 
$X$, which requires choosing a charge vector $\gamma \in H^{3}(X,\mathbb{Z})$ with
\begin{equation}
\gamma \in H^{3,0}(X_\psi) \oplus H^{0,3}(X_\psi) \label{eq:attractorEq}
\end{equation}
at an attractor point $\psi$ in moduli space.\footnote{The rich study of the relationship between arithmetic and special points in Calabi-Yau moduli space was originated by Moore in studies of the attractor equation \cite{moore:a&along}. At this juncture, we can describe an interesting connection between supersymmetric flux vacua and very special attractor points called ``rank two attractors" -- namely, that they specify the same points in CY moduli space in the
specific case of CY manifolds $X$ with $h^{2,1}(X)=1$.  Indeed, it was essentially pointed out in \cite{candelas:attractors} that, if $h^{2,1}=1$, one can define a supersymmetric flux vacuum at any point in moduli space admitting a rank-two attractor; note that this is distinct from the construction of \cite{moore:arithmeticLectures}, which associates a rank-two attractor to a nonsupersymmetric flux vacuum. In a rank-two attractor $X_\psi$, there are two integral charge vectors $\g_1,\g_2\in H^{3}(X_\psi,\bbZ)$ satisfying Eq. \ref{eq:attractorEq}. Taking their complement in $H^{3}(X_\psi,\bbZ)$, we easily find a two-dimensional lattice of integral fluxes satisfying Eq. \ref{eq:fh21plus12}. As long as $h^{2,1}=1$, we can straightforwardly tune $\t$ such that Eq. \ref{eq:g3h21} is satisfied for any two linearly independent elements $f$ and $h$ of this lattice,  and thus any rank-two attractor with $h^{2,1}=1$ defines a supersymmetric flux vacuum. The converse is also true: a supersymmetric flux vacuum with $h^{2,1}=1$ defines a rank two attractor, whose charge lattice is the orthogonal complement of the span of the fluxes. Moore \cite{moore:a&along} has conjectured that all rank-two attractors are defined over some number field. Given the above discussion, this conjecture also implies that all supersymmetric flux compactifications with $h^{2,1}=1$ are also defined over a number field. This is in sharp contrast to the case of generic $h^{2,1}$, where flux vacua can come in continuous families. }

 There is much more to say about supersymmetric flux compactifications, including rich connections to arithmetic previously studied in \cite{kachru:0411,dewolfe:11222}, but we will stop here; the further results 
will not be relevant to our main point.

\ignore{\subsection{Periods, Number Fields, and Supersymmetric Flux Compactifications}
\label{sec:cyclotomic}}

\section{Modularity From Supersymmetric Flux Compactifications}
\label{sec:main}
We can now relate supersymmetric flux compactifications to the modularity of the underlying Calabi-Yau. The basic idea is fairly simple: if a point $\phi$ in the complex structure moduli space of a Calabi-Yau threefold $X$ admits a supersymmetric flux compactification, the conditions in Eq. \ref{eq:fh21plus12} give us a two-dimensional subspace of the integral cohomology of $X_\phi$ with a pure Hodge structure of type (2,1) + (1,2). If $X_\phi$ is an algebraic variety defined over $\mathbb{Q}$, then by the Hodge conjecture this split induces a split of the \'etale cohomology of $X_\phi$ over a number field $K$. If this field $K$ is $\bbQ$, then we have recovered an \'etale split of the sort discussed in Section \ref{sec:modThree}, and $X_\phi$ is associated to a weight-two eigenform. (A similar argument was used to study the modularity of rank-two attractors in a one-parameter family of Calabi-Yau threefolds in \cite{candelas:attractors}.)

We will now make this argument precise. Consider a Calabi-Yau threefold $X$ with complex structure parameter $\phi$, and let $X_\phi$ be a fiber algebraically defined over $\bbQ$ that admits a supersymmetric flux compactification. Then the complex structure of $X_\phi$ has aligned such that there exist two integral cohomolgy elements $f,h\in H^{2,1}(X_\phi)\oplus H^{1,2}(X_\phi)$ which span a two dimensional subspace of $H^3(X_\phi,\bbQ)$. More precisely, denote by $H_{\text{flux}}$ the $\bbQ$-span of $f$ and $h$: \eq{H_{\text{flux}} = \bbQ{f}+\bbQ{h}.} Then we have a split \eq{H^3(X_\phi,\bbQ) = H_{\text{flux}} \oplus H_{\text{remainder}}.\label{eq:fluxQsplit}} The Hodge structure on $H^3(X_\phi,\bbQ)$ defines a Hodge structure on $H_{\text{flux}}$ with Hodge type (2,1)+(1,2), i.e. with no (3,0) or (0,3) components.

We would like to relate this split of singular cohomology to the discussion of Section \ref{sec:modThree}. However, that discussion was phrased in terms of \'etale cohomology, and so we need way to translate between these two cohomology theories. This is provided by the Hodge conjecture, which is introduced in detail in Appendix \ref{sec:wenzhe}. We will now briefly explain its role in our problem here, and importantly why the situation is not as nice as might have been hoped.

Let us back up slightly, and consider a variety $X/K$ over a number field $K$. The \'etale cohomologies $\het^i(X,\bbQ_\ell)$ of $X$ are representations of the Galois group $\gal(\bar{\bbQ}/K)$. The polynomials defining $X/K$ also define varieties $X/K'$ over all (finite) field extensions $K'$ of $K$. The \'etale cohomologies of these varieties over bigger number fields furnish representations of the corresponding Galois groups $\gal(\bbQb/K')$. Now let us assume that the singular cohomology of our original variety $X/K$ obeys \eq{H^3(X/K,\bbQ) = H' \oplus H'',\label{eq:hodgeSplit1}} with no assumption on the dimension of $H'$ or $H''$. Then it is shown in Appendix \ref{sec:hodgeconjecturereview} that there exists some finite field extension $K'$ of $K$ such that \eq{\het^3(X/K',\bbQ_\ell) = M' \oplus M'',\label{eq:hodgeSplit2}} where $M'$ has the same dimension and Hodge type of $H'$ (and similarly for $M''$). 

Now let us return to our supersymmetric flux compactification $X_\phi$ defined over $\bbQ$ and obeying Eq. \ref{eq:fluxQsplit}. Applying the Hodge conjecture, there exists a number field $K'$, i.e. a finite field extension of $\bbQ$, such that the variety $X_\phi/K'$ obeys \eq{\het^3(X_\phi/K',\bbQ_\ell) = M_{\text{flux}} \oplus M_{\text{remainder}} \label{eq:etaleSplitK'},}  where  $M_{\text{flux}}$ is a two-dimensional representation of $\gal(\bbQb/K')$ of Hodge type $(2,1) + (1,2)$. We can now investigate the modularity of $M_\flux$.

In the best, and simplest, case, $K'$ is just $\bbQ$, without any extension\footnote{The assumption in Section 1.3 of \cite{candelas:attractors} that the cycle $S$ is defined over $\bbQ$ amounts to assuming that $K'=\bbQ$.}. Then $X_\phi/K'$ is simply  $X_\phi$, and Eq. \ref{eq:etaleSplitK'} simply becomes \eq{\het^3(X_\phi,\bbQ_\ell) = M_{\text{flux}}\oplus M_{\text{remainder}},\label{eq:fluxEtaleSplit}} where now $M_\flux$ is a two-dimensional representation of $\gal(\bbQb/\bbQ)$ with Hodge type $(2,1) + (1,2)$, so that we can apply the discussion of Section \ref{sec:modularity} directly. Thus the Tate twist $M_{\text{flux}} \otimes \mathbb{Q}_\ell(1)$ of $M_\flux$ has Hodge type $(1,0) + (0,1)$ and is associated to a weight-two Hecke eigenform $f_2(\t) = \sum_n c_nq^n$. In particular, that means that, for good primes $p$, the $p$-th Fourier coefficient $c_p$ of $f_2(\t)$ will be present in the $\zeta$-function of $X_\phi$, according to Eq. \ref{eq:wt2factor}. 

Unfortunately, this is not the only possibility; it can also happen that, for some $X_\phi$, $K'$ is not simply $\bbQ$, but instead some larger number field. In this case, matters are much less clear. Of course, in the general spirit of the Langlands program we can (and do!) conjecture that $M_\flux$, or an appropriate Tate twist thereof, is modular. However, exactly what modularity means in this context is somewhat unclear. The precise nature of the automorphic form associated to $M_\flux$ depends strongly on the nature of $K'$. For instance, if $K'$ is a real quadratic field then we expect $M_\flux$ to be associated to a Hilbert modular form \cite{consani2000geometry,schutt:hilbert,straten:hilbert}. On the other hand, for more general fields it is not known what sort of automorphic object we expect to find; see e.g. \cite{taylor:CMfields} for recent progress on the study of modularity for elliptic curves over more complicated number fields. Even if $K'$ is a more general field, however, it is still sometimes possible to find  ordinary modular forms in the $\zeta$-functions of $X_\phi$\footnote{This intuition is supported by the Chebotarev density theorem \cite{densityTheorem}, which states that, regardless of $K'$, $R(t)$ should factor at an infinite number of primes.}. This can be seen for example in the $\bbQ[\sqrt{17}]$ examples of \cite{candelas:attractors}.

So far we have restricted ourselves to supersymmetric flux vacua defined over $\bbQ$. However, this is not really necessary; the same argument goes through almost exactly if we start with a more general number field. A fiber $X_\phi/K$ defined over a number field $K$ is still described by a choice of $G_3$ satisfying Eq. \ref{eq:g3h21}, and thus still satisfies Eqs. \ref{eq:fh21plus12} and \ref{eq:fluxQsplit}. Then by Hodge there exists some field extension $K'$ of $K$ over which Eq. \ref{eq:etaleSplitK'} is satisfied, and $M_\flux$ is a two-dimensional representation of $\gal(\bbQb/K')$; we thus conjecture it to be modular. This is exactly the same as the $K = \bbQ, K' \neq \bbQ$ situation, with one important caveat: $K'$ is always at least as big as $K$, so if $K$ is bigger than $\bbQ$ then so is $K'$, and the results of Section \ref{sec:modularity} will never apply. We have thus restricted ourselves to $\bbQ$ only for simplicity.

Before we move on, we will make several comments. First, let us consider the special case $h^{2,1}=1$. A supersymmetric flux compactification over $\bbQ$ is weight-two modular, but as discussed above is also a rank-two attractor, and thus also weight-four modular. Thus, the weight-two modularity observed in the rank-two attractor examples of \cite{candelas:attractors} is related to the presence of a supersymmetric flux compactification.

Next, we note that the converse of our result is not necessarily true: not any threefold defined over $\bbQ$ that is weight-two modular is a supersymmetric flux vacuum. Any weight-two modular threefold obeys Eq. \ref{eq:fluxEtaleSplit}. Although Eq. \ref{eq:fluxQsplit} does not imply Eq. \ref{eq:fluxEtaleSplit}, any fiber $X_\phi$ satisfying Eq. \ref{eq:fluxEtaleSplit} also obeys Eq. \ref{eq:fluxQsplit}. Any such $X_\phi$ admits integral fluxes $f,h$ satisfying Eq. \ref{eq:fh21plus12}. However, this is not enough to imply Eq. \ref{eq:g3h21}, which we can only satisfy if the projections of $f$ and $h$ onto $H^{1,2}(X_\phi)$ are collinear. This will not be true in general, but is always true if $h^{2,1}=1$. Thus, weight-two modularity only necessarily implies the presence of a supersymmetric flux compactification if  $h^{2,1}=1$.

Finally, we note briefly that the same argument could have been obtained from Eq. \ref{eq:g3h21} directly, without recourse to Eq. \ref{eq:fh21plus12}. However, then we would be considering the span of $G_3$ and $\bar{G}_3$, which are not in integral cohomology but instead in the singular cohomology over the ring $\bbZ[\t]$. This complicates the argument, which should then be phrased in a more general context, but a similar conclusion can be reached.

\section{An Example: The Octic in $\bbP(1,1,2,2,2)$}
\label{sec:octic}

We have argued that the Hodge conjecture implies that at least some supersymmetric flux compactifications defined over $\bbQ$, namely those with $K'=\bbQ$, are modular, and associated to a weight-two eigenform. Unfortunately, we do not know of a simple way to identify $K'$; it must be understood separately for each $X_\phi$, i.e. for each rational point in the moduli space of a threefold $X$ admitting a supersymmetric flux compactification. We will therefore spend the rest of this paper investigating the following question in a rich example:   

\begin{quote}\begin{centering}\textbf{Let $X_{\phi}$ be a CY threefold defined over $\bbQ$ admitting a supersymmetric flux compactification. Does there exist a  weight-two Hecke eigenform $f_2(\t) = \sum_n c_nq^n$ associated to the middle cohomology of $X_{\phi}$, in the sense of Eq. \ref{eq:wt2factor}?}\end{centering}\end{quote}

Our main example is the octic hypersurface in $\bbP(1,1,2,2,2)$, studied in e.g. \cite{candelas:2param1}; generic fibers in this family have Hodge numbers $h^{1,1}=2$, $h^{2,1} = 86$. This model is particularly suited for our purposes, because it is known to admit supersymmetric flux compactifications \cite{dewolfe:11222} and its $\zeta$-functions have been computed for small prime numbers \cite{kadir:octic}. We want to check whether, at rational points  admitting supersymmetric flux compactifications, the Fourier coefficients of weight-two eigenforms are present in quadratic factors of the numerators of the zeta functions.

Flux vacua in this model are easiest to study on the mirror locus, where the threefolds $X_{\psi,\phi}$ are constructed from the vanishing locus of the polynomial \eq{x_1^8+x_2^8+x_3^4+x_4^4+x_5^4 - 8\psi x_1x_2x_3x_4x_5 - 2\phi x_1^4x_2^4 = 0.\label{eq:octic}} It was proven in \cite{dewolfe:11222} that this model admits a continuum of flux vacua. In particular, supersymmetric flux vacua can be found on the entire $\psi=0$ locus. Thus, we expect the points of the form $\psi = 0, \phi\in\bbQ$ to be associated to weight-two Hecke eigenforms.

To diagnose modularity, we will look for factorizations of the $\zeta$-function. By the argument of \cite{candelas:finite2}, it will be sufficient for us to consider only the numerator $R_0(t)$ of the $\zeta$-function of the mirror family, which has order $2+2h^{1,1}=6$. See the Appendix \ref{sec:mirror} for more details about this argument. Notice that $R_0(t)$ is the characteristic polynomial of the Frobenius element acting on the third \'etale cohomology of the mirror of $X_{\psi,\phi}$. This factor was computed for all $\psi,\phi\in\bbQ$ and for all primes less than 19 in \cite{kadir:octic}. In each case, $R_0(t)$ is a sextic polynomial in $t$, and has one or more quadratic factors. If $X_\phi$ is modular, then one quadratic factor will, by Eq. \ref{eq:wt2factor}, be related to the $p$-th Fourier coefficient of the associated weight-two eigenform. Thus, for each $p$, we can read off one or more possible Fourier coefficients $c_p^\zeta$. Cross-referencing those possible values with a table of weight-two Hecke eigenforms, such as those found in \cite{interactiveTable,meyer:table2}, must yield a unique eigenform if $X_\phi$ is modular. We will check this for several points in moduli space, and find modularity in all examples.

The most obvious point to check is the Gepner point, $\psi=\phi=0$. This is an especially nice point arithmetically, as the axiodilaton $\phi$ lives in a cyclotomic field \cite{dewolfe:11222}. At this point, the only bad prime is $p=2$, as can easily be verified by differentiating Eq. \ref{eq:octic}. For all other primes less than 19, the mirror numerators $R_0(t)$ have been computed in \cite{kadir:octic}. These are listed in Table \ref{tab:phi0}. We are interested in finding a split of the form given in Eq. \ref{eq:wt2factor}. Comparing Eq. \ref{eq:wt2factor} with Table \ref{tab:phi0}, we need \eq{c_5 = 2, c_{13} = -6, c_{17} = 2,} with $c_p$ vanishing for all other primes less than 19.  Consulting online tables of wieght-two modular forms \cite{interactiveTable,meyer:table2} shows that there is a unique weight-two Hecke eigenform with these coefficients, which happens to be the unique wieght-two newform for $\Gamma_0(64)$, called 64.2.a.a in \cite{interactiveTable}. It is very interesting to notice that the level of this newform, i.e. 64, is a power of the bad prime $p=2$. Thus, we find strong numerical evidence that the Gepner point is associated to a weight-two Hecke eigenform, as expected\footnote{This eigenform has been associated to the octic hypersurface before. The octic is ruled by a genus-three surface \cite{kadir:octic,rolf:genus3}, and so in the spirit of \cite{verrill:ruled} is weight-two modular. It turns out that the modular form associated to this genus-three surface is again 64.2.a.a, and so its Fourier coefficients appear in a so-called ``exceptional factor" $R_{\text{exceptional}}(t)$ in the numerator of the $\zeta$-function \cite{kadir:octic}. However, this is quite different than its appearance here. In particular, $R_{\text{exceptional}}$ is independent of $\psi$ and $\phi$, whereas 64.2.a.a only appears in $R_0(t)$ at the Gepner point $\psi=\phi=0$.}.  

\begin{table}[h]
\begin{centering}
\begin{tabular}{|c|c|c|c|}\hline
$p$ & $R_{0}(t)$  & $c_p^\zeta$ & $c_p$\\\hline
3 & $\left(1 + 3^3 t^2\right)\left(1 - 18t^2 + 3^6t^4\right)$ & 0 & 0 \\\hline
5 & $\left(1 - 10t + 5^3t^2\right)\left(1 - 70t^2 + 5^6t^4\right)$ & 2 & 2 \\\hline
7 & $\left(1+7^3t^2\right)\left(1+686t^2+7^6t^4\right)$ & 0 &0 \\\hline
11 & $\left(1 + 11^3t^2\right)\left(1+1694t^2+11^6t^4\right)$ & 0 & 0 \\\hline
13 & $\left(1 + 78t + 13^3t^2\right)\left(1 -3094t^2 + 13^6t^4\right)$ &-6 & -6 \\\hline
17 & $\left(1 - 34t + 17^3t^2\right)\left(1 + 180t + 15878t^2 + 17^3180t^3 + 17^6t^4\right)$ & 2 & 2\\\hline
\end{tabular}
\caption{Contributions to the $\zeta$-function of the mirror octic at the point $\psi=\phi=0$. These are listed in Appendix B of \cite{kadir:octic}. From the form of $R_0(t)$ for each $p$, we can find the Fourier coefficient $c_p^\zeta$ that a weight-two eigenform must have to satisfy Eq. \ref{eq:wt2factor}. We list these alongside the Fourier coefficients of the weight-two eigenform 64.2.a.a \cite{interactiveTable}. These match for each $p$, and thus we conclude that the mirror octic at the point $\psi=\phi=0$ is modular.}
\label{tab:phi0}
\end{centering}
\end{table}

We have also found evidence for modularity for several other rational points along the $\psi=0$ locus. These points are summarized in Table \ref{tab:summary}. For each choice of $\phi$, we list the primes $p$ of bad reduction, and the appropriate weight-two Hecke eigenform, following the naming conventions of \cite{interactiveTable}; this labeling indicates the level $N$ of the eigenform, which is divisible only by the bad primes of $X_{\psi,\phi}$. Analogs of Table \ref{tab:phi0} are provided for each point in Appendix \ref{sec:data}, including both the $\zeta$-function numerators, computed in \cite{kadir:octic}, and the Fourier coefficients of the appropriate Hecke eigenforms. We see that each mirror $\zeta$ numerator has a quadratic factor encoding the appropriate Fourier coefficient. These are highly nontrivial checks, and we consider these results to be very strong evidence for the modularity of these rational supersymmetric flux vacua. 

\begin{table}[h]
\begin{centering}
\begin{tabular}{|c|c|c|}\hline
$\phi$ & Bad primes & Modular form \\\hline
$1/2$ & 2, 3 & 48.2.a.a \\\hline
$3/5$ & 2, 5 & 400.2.a.e \\\hline 
11/8 & 2, 3, 19 & 912.2.a.b \\\hline
$2$ & 2, 3 & 192.2.a.a \\\hline
3 & 2 & 32.2.a.a \\\hline
7 & 2, 3 & 24.2.a.a \\\hline
$9$ & 2, 5 & 40.2.a.a\\\hline
\end{tabular}
\caption{Summary of modularity results for rational values of $\phi$. For each $\phi$, we list the primes of bad reduction of $X_\phi$, and the weight-two Hecke eigenform associated to that point in moduli space, following the nomenclature of \cite{interactiveTable}. For each label, the integer gives the level of the modular form, which we note is divisible only by the bad primes. Data supporting these results is given in Appendix \ref{sec:data}.}
\label{tab:summary}
\end{centering}
\end{table}

We saw in Section \ref{sec:main} that it is also possible that rational supersymmetric flux vacua are not associated to an ordinary Hecke eigenform, but instead to some more general automorphic form. However, we have not been able to find any examples of such points. In addition to the points listed in Table \ref{tab:summary}, we have also studied many other points, and in all cases have found that the $\zeta$-functions are compatible with modularity\footnote{The points listed in Table \ref{tab:summary} are those for which we were able to pin down a unique newform. In the other examples we studied, the first few primes only provided enough data to restrict our consideration to a small number of newforms, instead of a unique one.}. This suggests that, at least in this example, the $K'=\bbQ$ case associated to an ordinary newform is relatively common, and we are cautiously optimistic that it is in some sense the ``generic" case.

\section{Conclusion}
\label{sec:conclusion}
In this paper, we have studied the modularity of supersymmetric flux compactifications over $\bbQ$.
However, we have only scractched the surface, and the story is far from complete. From the broader perspective of the Langlands program, we expect every \ignore{algebraically defined} Calabi-Yau threefold defined over a number field to be modular in a more general sense. While we are still far from understanding the modularity of more general Calabi-Yau threefolds, it is tempting to hope that string theory might be relevant.  In similar spirit, a complementary study of the modularity of rank-two attractor points in a one-parameter family Calabi-Yau threrefolds recently appeared in \cite{candelas:attractors}.

Even in the limited context of threefolds with \'etale splits, there is still work to be done. Here and in \cite{candelas:attractors}, the physical criteria were defined in terms of a splitting the singular cohomology of a special fiber $X_*$, and then this split was translated into a split of the \'etale cohomology by means of the Hodge conjecture. Doing so introduces the extra complication of the splitting field $K'$. While in the simple case $K' = \bbQ$ we understand the resulting modularity very well, for more general $K'$ we are comparatively ignorant. However, in examples we were unable to find any cases in which this added complication is relevant. Therefore, research into the nature of $K'$ is a natural line of inquiry. In particular, can  necessary and sufficient criteria on a fiber $X_\phi$ be found to ensure that $K' = \bbQ$? In the opposite direction, an example of a threefold defined over $\bbQ$ and admitting a supersymmetric flux compactification for which $K'\neq\bbQ$ would also be extremely interesting.

In this paper, we have only explicitly studied the octic hypersurface in $\bbP(1,1,2,2,2)$. There is no principled reason for this; we simply wanted an example where the $\zeta$-functions are already known. However, many other examples of supersymmetric flux vacua are known, including the sextic in $\bbP(1,1,1,1,2)$ \cite{kachru:0411} and several families with $h^{1,1}=2$ \cite{dewolfe:11222,kachru:fluxThreefolds} and $h^{1,1}=3$ \cite{dewolfe:11222}. If one could compute the $\zeta$-functions of these families, the analysis of Section \ref{sec:octic} could easily be repeated; we expect that similar results would be found. 

One can also imagine running this argument in reverse, and using known modularity results to find new supersymmetric flux compactifications. As discussed above, this is only guaranteed to work if $h^{2,1}=1$. Fortunately, several examples of threefolds with $h^{2,1}=1$ associated to weight-two modular forms are known. These include double octics \cite{meyer:book}, elliptically ruled surfaces \cite{meyer:book,verrill:ruled}, at least one Horrocks-Mumford quintic \cite{lee:quinticModular}, and at least one toric variety \cite{verrill:a4}. These are thus novel examples of supersymmetric flux compatifications. Further examples are provided by the rank-two attractors with $h^{2,1}=1$ identified in \cite{candelas:attractors}.

In the models with $h^{2,1}>1$ discussed in \cite{dewolfe:11222}, including the octic studied here, supersymmetric flux vacua can come in continuous families. Thus, in addition to the rational points whose modularity is the subject of this paper, these families also have fibers over number fields bigger than $\bbQ$. As mentioned above, most of the argument of Section \ref{sec:main} goes through for such fibers, and it seems likely that these fibers will also be modular, albeit for much more complicated automorphic forms. Thus, we suggest that these models might be a fertile ground for analyses along the lines of \cite{consani2000geometry,schutt:hilbert,straten:hilbert} to enlarge the number of known examples of modular threefolds for other number fields. \ignore{However, even for elliptic curves, not much is known about modularity over even quadratic imaginary fields; see e.g. \cite{taylor:CMfields} for recent progress. }

An obvious question from the string theory point of view is whether the automorphic forms associated to special points in Calabi-Yau moduli space admit a physical interpretation. \footnote{For elliptic curves, some discussion of this point recently appeared in \cite{kondo:worldsheetLFunctions,kondo:worldsheetLFunctions2}. For progress on threefolds see e.g. \cite{Kadir:2010dh}.}  Finding a direct physical interpretation of the Fourier coefficients of the associated modular forms (or the point counts) is a clear goal for future work.

\section*{Acknowledgements}
This research of SK was supported by the NSF under grant PHY-1720397 and by a Simons Investigator Award.
RN thanks C. Hewett, B. Rayhaun, and Z. Yao for helpful discussions, and is funded by NSF Fellowship DGE-1656518 and an EDGE grant from Stanford University. WY is supported by the Stanford Institute for Theoretical Physics.

\appendices

\section{$\zeta$-Functions and the Greene-Plesser Mirror}
\label{sec:mirror}
In this Appendix, we will review the relationship, described in \cite{candelas:finite2}, between the $\zeta$-function of a projective hypersurface $X$ and its Greene-Plesser mirror $Y$. Fix some weighted projective space $\bbP(k_1,k_2,k_3,k_4,k_5)$, where each of the weights $k_i$ divides the total weight $K=\sum_{i=1}^5k_i$. We can write down a projective CY hypersurface $X$ in this space as \eq{\sum_{i=1}^5 x_i^{K/k_i} - \phi_{\bf{A}}x^{\bf{A}} = 0\label{eq:generalHypersurface},} where the $\phi_{\bf{A}}$ are the $h^{2,1}(X)$ complex structure moduli of $X$ and the $x^{\bf{A}}$ are monomials with total weight $K$. Following \cite{candelas:finite1,candelas:finite2,candelas:dwork}, we have introduced a multi-index ${\bf{A}}=(a_1,a_2,a_3,a_4,a_5)$; the notation $x^{\bf{A}}$ should be taken to mean \eq{x^{\bf{A}} = \prod_{i=1}^5 x_i^{a_i}.} The $\zeta$-function of $X$ takes the form \eq{\zeta(X,t) = \frac{R(t)}{D(t)},} and we are interested primarily in the numerator $R(t)$, which is a polynomial of degree $b^3(X) = 2+2h^{2,1}(X).$

Now consider the mirror $Y$ of $X$. The Greene-Plesser construction of $Y$ realizes it as a quotient of $X$ \cite{GreenePlesser}. In general, only some subset of the $x^{\bf{A}}$'s will survive this quotient, i.e. we can write $Y$ as the vanishing of a polynomial \eq{\sum_{i=1}^5 y_i^{\tilde{K}/\tilde{k}_i} - \psi_{\tilde{\bf{A}}} y^{\tilde{\bf{A}}} = 0.} Write the $\zeta$-function of $Y$ as \eq{\zeta(Y,t) = \frac{\bR(t)}{\bD(t)}.} Now, the order of $\bR(Y,t)$ is $2+2h^{2,1}(Y)$, which by mirror symmetry is equal to $2+2h^{1,1}(X)$.  

As is familiar, it is convenient to separate out from Eq. \ref{eq:generalHypersurface} the subset of monomials that survive the quotient, i.e. we usually write $X$ as the vanishing of \eq{\sum_{i=1}^5 x_i^{K/k_i} -\psi_{\tilde{\bf{A}}}x^{\tilde{\bf{A}}} -  \phi_{\bf{A}}x^{\bf{A}} = 0\label{eq:generalHypersurfacemirror},} where now $\tilde{\bf{A}}$ runs over $h^{1,1}(X)$ monomials and $\bf{A}$ runs over $h^{2,1}(X)-h^{1,1}(X)$ monomials. Now comes the crucial point. It was observed in \cite{candelas:finite1,candelas:finite2,candelas:dwork} that \eq{R(t) = \bR(t) \cdot R_1(t).} Those references focused specifically on the quintic, but the same argument should apply much more generally; the same ideas were applied to the octic in $\bbP(1,1,2,2,2)$ in  \cite{kadir:octic}. In the literature on the arithmetic of flux compactifications, and in particular \cite{kachru:0411,dewolfe:11222,kachru:fluxThreefolds}, flux compactifications are usually implicitly studied on the mirror, so we are really interested in studying the modular properties of $\bR(t)$. In the language of $X$, this means that we are free to neglect the $R_1(t)$ term, and instead only focus on the (comparatively simple) factor $\bR(t)$. This is the point of view we will take throughout the main text.

\section{Weil Cohomology Theories, Motives, and the Hodge Conjecture}
\label{sec:wenzhe}
\input{wenzheAppendix.tex}

\section{Modularity of the Octic at Rational Points}
\label{sec:data}
In Table \ref{tab:summary}, we stated modularity results for several rational values of $\phi$. In this Appendix, we provide tables justifying these claims. For each choice of $\phi$, we find several good primes $p$. We fix a good  prime $p$, and reduce $\phi$ mod $p$. We then look up the corresponding $\zeta$-factor $R_0(t)$ for the mirror octic in the tables of \cite{kadir:octic}. In all cases we consider, $R_0(t)$ is a sextic in $t$, and factors over $\bbZ$ into either the product of a quartic and a quadratic or the product of three quadratics. Comparing each quadratic factor to Eq. \ref{eq:wt2factor}, we find one or more possible values $c_p^\zeta$. For $X_\phi$ to be modular, there must exist a weight-two eigenform $f_2$ whose $p$-th coefficient matches one of the $c_p^\zeta$. Comparing with a table of eigenforms, such as \cite{interactiveTable}, we have found a unique such eigenform for each of the five points considered here. Thus, we find strong numerical evidence of modularity in each example.

We now provide the data backing up these conclusions. For each value of $\phi$ listed in Table \ref{tab:summary}, we present a table showing the reductions of $\phi$ for each good prime, as well as the appropriate form of $R_0$. From these, we find values of $c_p^\zeta$, and from the $c_p^\zeta$ we find a weight-two Hecke eigenform, whose Fourier coefficients are also listed. Following \cite{kadir:octic}, we sometimes use the shorthand \subeqs{(a)_2 &= 1+at+p^3t^2 \\ (a,b)_4 &= 1+at+bt^2+ap^3t^2+p^6t^4.}

\subsection{$\phi=1/2$}
\begin{table}[H]
\begin{centering}
\begin{tabular}{|c|c|c|c|c|}\hline
$p$ & $\phi\mod{p}$ & $R_0(t)$ & $c_p^\zeta$ & $c_p$ \\\hline
5 & 3 & $(1+4t+125t^2)(1+10t+125t^2)^2$ & -2 & -2 \\\hline
7 & 4 &  $(1+343t^2)(1+28t+343t^2)(1-28t+343t^2)$ & $\pm$4 or 0 & 0 \\\hline
11 & 6 &  $(44)_2(0,-1210)_4 $ & -4 & -4 \\\hline
13 & 7 &  $(-6\times13)_2(2\times13)_2(-18)_2 $ &-2 or 6 & -2  \\\hline
17 & 9 & $(-136)_2(-34)_2(94)_2 $ & 2 or 8 &  2  \\\hline
\end{tabular}
\caption{Zeta functions for the mirror octic at the point $\psi=0,$ $\phi=1/2$, as computed in \cite{kadir:octic}. From each polynomial,  we have extracted the possible coefficients $c_p^\zeta$ that a weight-two eigenform must have to satisfy Eq. \ref{eq:wt2factor}.  These are displayed alongside the Fourier coefficients $c_p$ of the weight-two Hecke eigenform 48.2.a.a \cite{interactiveTable}. Note that for each $p$, the coefficients match exactly. Thus, we have found evidence for modularity.}
\label{tab:zeta1/2}
\end{centering}
\end{table}

\subsection{$\phi=3/5$}
\begin{table}[H]
\begin{centering}
\begin{tabular}{|c|c|c|c|c|}\hline
$p$ & $\phi\mod{p}$ & $R_0(t)$ & $c_p^\zeta$  & $c_p$\\\hline
3 & 0 & $(1+27t^2)(1-18t^2+729t^4)$ & 0 & 0 \\\hline
7 & 2 & $ (1+28t+343t^2)(0,294)_4$ & -4 & -4\\\hline
11 & 5 & $(44)_2(0,-1210)_4 $ & -4 & -4 \\\hline
13 & 11 & $(-2\times13)_2,(48,3770)_4 $ &  2 &  2  \\\hline
17 & 4 & $(-136)_2(34)_2(104)_2 $ & -2 or 8 & -2 \\\hline
\end{tabular}
\caption{Zeta functions for the mirror octic at the point $\psi=0,$ $\phi=3/5$, as computed in \cite{kadir:octic}. From each polynomial,  we have extracted the possible  coefficients $c_p^\zeta$ that a weight-two eigenform must have to satisfy Eq. \ref{eq:wt2factor}.  These are displayed alongside the Fourier coefficients $c_p$ of the weight-two Hecke eigenform 400.2.a.e \cite{interactiveTable}. Note that for each $p$, the coefficients match exactly. Thus, we have found evidence for modularity.}
\label{tab:zeta3/5}
\end{centering}
\end{table}

\subsection{$\phi = 11/8$} 
\begin{table}[H]
\begin{centering}
\begin{tabular}{|c|c|c|c|c|}\hline
5 & 2 & $(10)_2(-10)_2(-4)_2$ &  $\pm$ 2 & -2 \\\hline
7 & 4 & $(0)_2(28)_2(-28)_2$ & 0 or $\pm$ 4 & 0 \\\hline
11 & 0 & $(0)_2(0,1694)_4$ & 0 & 0 \\\hline
13 & 3 & $(-6\times13)_2(102,2\times3\times79\times13)_4$ & 6 & 6\\\hline
17 & 12 & $(102)_2(-46,8738)_4$ & -6 & -6 \\\hline
\end{tabular}
\caption{$\zeta$-function numerators for the mirror octic at the point $\psi=0,$ $\phi=11/8$, as computed in \cite{kadir:octic}. From each polynomial,  we have extracted the possible coefficients $c_p^\zeta$ that a weight-two eigenform must have to satisfy Eq. \ref{eq:wt2factor}.  These are displayed alongside the Fourier coefficients $c_p$ of the weight-two Hecke eigenform 912.2.a.b \cite{interactiveTable}. Note that for each $p$, the coefficients match exactly. Thus, we have found evidence for modularity.}
\label{tab:zeta11/8}
\end{centering}
\end{table}

\subsection{$\phi = 2$} 
\begin{table}[H]
\begin{centering}
\begin{tabular}{|c|c|c|c|c|}\hline
$p$ & $\phi\mod{p}$ & $R_0(t)$ & $c_p^\zeta$ & $c_p$ \\\hline
5 & 2 & $(1-10t+125t^2)(1+10t+125t^2)(1-4t+125t^2)$ & 2 or -2 & -2 \\\hline
7 & 2 & $(1+28t+343t^2)(1+294t^2 + 7^6t^4)$ & -4 & -4 \\\hline
11 & 2 & $(-44)_2(44)_2^2$ & 4  or -4 & -4 \\\hline
13 & 2 & $(-26)_2(48,3770)_4$ & 2 & 2 \\\hline
17 & 2 & $(102)_2(-14,-7582)_2$ & -6 & -6 \\\hline
\end{tabular}
\caption{$\zeta$-function numerators for the mirror octic at the point $\psi=0,$ $\phi=2$, as computed in \cite{kadir:octic}. From each polynomial,  we have extracted the possible coefficients $c_p^\zeta$ that a weight-two eigenform must have to satisfy Eq. \ref{eq:wt2factor}.  These are displayed alongside the Fourier coefficients $c_p$ of the weight-two Hecke eigenform 192.2.a.a \cite{interactiveTable}. Note that for each $p$, the coefficients match exactly. Thus, we have found evidence for modularity.}
\label{tab:zeta2}
\end{centering}
\end{table}

\subsection{$\phi=3$}

\begin{table}[H]
\begin{centering}
\begin{tabular}{|c|c|c|c|c|}\hline
$p$ & $\phi\mod{p}$ & $R_0(t)$ & $c_p^\zeta$ & $c_p$ \\\hline
3 & 0 & $(0)_2(0,-18)_4$ & 0 & 0 \\\hline
5 & 3 & $(4)_2(10)_2^2$ &-2 & -2 \\\hline
7 & 3 & $(0)_2(28)_2(-28)_2$ & 0 or $\pm$4 & 0 \\\hline
11 & 3 & $(0)_2(0,1694)_4$  & 0 & 0 \\\hline
13 & 3 & $(-6\times13)_2((102,628524)_4$ & 6 & 6\\\hline
17 & 3 & $(-34)_2(0,-5474)_4$ & 2 & 2 \\\hline\end{tabular}
\caption{Zeta functions for the mirror octic at the point $\psi=0,$ $\phi=3$, as computed in \cite{kadir:octic}. From each polynomial,  we have extracted the possible coefficients $c_p^\zeta$ that a weight-two eigenform must have to satisfy Eq. \ref{eq:wt2factor}.  These are displayed alongside the Fourier coefficients $c_p$ of the weight-two Hecke eigenform 32.2.a.a \cite{interactiveTable}. Note that for each $p$, the coefficients match exactly. Thus, we have found evidence for modularity.}
\label{tab:zeta3}
\end{centering}
\end{table}

\subsection{$\phi=7$}

\begin{table}[H]
\begin{centering}
\begin{tabular}{|c|c|c|c|c|}\hline
$p$ & $\phi\mod{p}$ & $R_0(t)$ & $c_p^\zeta$ & $c_p$ \\\hline
5 & 2 & $(-10)_2(10)_2(-4)_2$ & 0 or $\pm$ 2 & -2 \\\hline
7 & 0 & $(0)_2^3$ & 0 & 0\\\hline
11 & 7 & $(0)_2^2(-44)_2$ & 0 or 4 & 4 \\\hline
13 & 7 & $(-6\times13)_2(2\times13)_2(-18)_2$ & 6 or -2 & -2 \\\hline
17 & 7 & $(-34)_2^2(104)_2$ & 2 & 2 \\\hline
\end{tabular}
\caption{$\zeta$-function numerators for the mirror octic at the point $\psi=0,$ $\phi=7$, as computed in \cite{kadir:octic}. From each polynomial,  we have extracted the possible coefficients $c_p^\zeta$ that a weight-two eigenform must have to satisfy Eq. \ref{eq:wt2factor}.  These are displayed alongside the Fourier coefficients $c_p$ of the weight-two Hecke eigenform 24.2.a.a \cite{interactiveTable}. Note that for each $p$, the coefficients match exactly. Thus, we have found evidence for modularity.}
\label{tab:zeta7}
\end{centering}
\end{table}

\subsection{$\phi=9$}

\begin{table}[H]
\begin{centering}
\begin{tabular}{|c|c|c|c|c|}\hline
$p$ & $\phi\mod{p}$ & $R_0(t)$ & $c_p^\zeta$ & $c_p$ \\\hline
3 & 0 & $\left(1 + 3^3 t^2\right)\left(1 - 18t^2 + 3^6t^4\right)$ & 0 & 0 \\\hline
7 & 2 & $(1+28t+343t^2)(0,294)_4$ & -4 & -4 \\\hline
11 & 9 & $(-44)_2(44)_2^2$ & 4 or -4 & 4 \\\hline
13 & 9 & $(26)_2(42,-78)_4$ & -2 & -2 \\\hline
17 & 9 & $(-136)_2(-34)_2(94)_2$ & 2 or 8 &  2 \\\hline
\end{tabular}
\caption{$\zeta$-function numerators for the mirror octic at the point $\psi=0,$ $\phi=9$, as computed in \cite{kadir:octic}. From each polynomial,  we have extracted the possible coefficients $c_p^\zeta$ that a weight-two eigenform must have to satisfy Eq. \ref{eq:wt2factor}.  These are displayed alongside the Fourier coefficients $c_p$ of the weight-two Hecke eigenform 40.2.a.a \cite{interactiveTable}. Note that for each $p$, the coefficients match exactly. Thus, we have found evidence for modularity.}
\label{tab:zeta9}
\end{centering}
\end{table}

\bibliographystyle{hunsrt}
\bibliography{flux}

\end{document}

%% file: wenzheAppendix.tex
In this Appendix we will review some background material from algebraic geometry. We will first briefly introduce Galois representation and properties of $L$-functions. We will then review the construction of pure motives, before moving on to a motivic formulation of the Hodge conjecture that is critical in the study of the modularity of supersymmetric flux compactifications.

\subsection{The Absolute Galois Group} \label{AbsoluteGaloisgroup}

\noindent First we will briefly introduce the absolute Galois group $\text{Gal}(\overline{\mathbb{Q}}/\mathbb{Q})$; all of the information presented here can be found in e.g. \cite{Neukirch, SerreLF, Taylor}. Let $K$ be a number field, i.e. a finite extension of the field of rational numbers $\mathbb{Q}$. An element of $K$ is called an algebraic integer if it is a solution to an integral monic polynomial
\begin{equation}
x^n+a_{n-1}x^{n-1}+ \cdots +a_1 x+a_0=0, ~\text{with}~a_i \in \mathbb{Z}.
\end{equation}
The set of all algebraic integers forms a subring of $K$ that will be denoted by $\mathcal{O}_K$, and it includes $\mathbb{Z}$ as a subring. The ring $\mathcal{O}_K$ is a Dedekind domain, i.e. it is an integral domain in which every non-zero proper ideal has a unique factorization as a product of prime ideals. In particular the principal ideal $(p) \subset \mathcal{O}_K$, generated by a prime number $p \in \mathbb{Z}$, has a factorization into prime ideals of the form 
\begin{equation} \label{eq:factorizationofprimep}
(p)=\mathfrak{P}_1^{e_1}\cdots \mathfrak{P}_g^{e_g},~\text{with}~\mathfrak{P}_i \neq \mathfrak{P}_j~\text{when}~i \neq j,
\end{equation}
where the integers $e_i\ge1$ are called the ramification indices. Every nonzero prime ideal of $\mathcal{O}_K$ is maximal, and therefore $\mathcal{O}_K/\mathfrak{P}_i$ is a field that is a finite extension of the finite field $\mathbb{F}_p:=\mathbb{Z}/p\mathbb{Z}$ \cite{SerreLF}. The degree 
\begin{equation}
f(\mathfrak{P}_i/p):=[\mathcal{O}_K/\mathfrak{P}_i:\mathbb{F}_p]
\end{equation}
of this extension is called the residue class degree. We have a relation
\begin{equation}
\sum_{i=1}^g e_i\,f(\mathfrak{P}_i/p) =[K:\mathbb{Q}]~.
\end{equation}
We say $p$ is ramified if any of the $e_i$ is greater than one, and unramified otherwise. If a prime ideal $\mathfrak{P}$ of $\mathcal{O}_K$ occurs in the factorization \eqref{eq:factorizationofprimep}, we say $\mathfrak{P}$ divides $p$ and write $e(\mathfrak{P}/p)$ (resp. $f(\mathfrak{P}/p)$) for its ramification index (resp. residue class degree $[\mathcal{O}_K/\mathfrak{P}:\mathbb{F}_p]$). The prime number $p$ is said to split in $\mathcal{O}_K$ if $e_i=f(\mathfrak{P}_i/p)=1$ for every $i$ in the factorization \eqref{eq:factorizationofprimep}, while $p$ is said to be inert in $\mathcal{O}_K$ if $(p)$ is a prime ideal of $\mathcal{O}_K$.

Let $\overline{\mathbb{Q}}$ be the algebraic closure of $\mathbb{Q}$ and $\overline{\mathbb{Z}}$ be the subring of $\overline{\mathbb{Q}}$ that consists of all algebraic integers. The absolute Galois group $\text{Gal}(\overline{\mathbb{Q}}/\mathbb{Q})$ is the group of automorphisms of $\overline{\mathbb{Q}}$ that preserve $\bbQ$ pointwise. It is a profinite group that is given by the inverse limit
\begin{equation}
\text{Gal}(\overline{\mathbb{Q}}/\mathbb{Q})=\varprojlim_L \text{Gal}(L/\mathbb{Q}),
\end{equation}  
where $L$ runs over all the finite Galois extensions of $\mathbb{Q}$. This inverse limit endows $\text{Gal}(\overline{\mathbb{Q}}/\mathbb{Q})$ with a topology called the Krull topology, which is compact, Hausdorff and totally disconnected. For a prime number $p \in \mathbb{Z}$, let $\mathbb{Q}_p$ be the field of $p$-adic numbers, whose algebraic closure will be denoted by $\overline{\mathbb{Q}}_p$. From the constructions of completions with respect to nonarchimedean norms, an embedding $\overline{\mathbb{Q}} \hookrightarrow \overline{\mathbb{Q}}_p$ is uniquely determined by the choice of a prime ideal $\overline{p} \subset \overline{\mathbb{Z}}$ such that $\overline{p} \cap \mathbb{Z}=(p)$ \cite{SerreLF}. The element of $\text{Gal}(\overline{\mathbb{Q}}/\mathbb{Q})$ that admits an extension to a continuous automorphism of $\overline{\mathbb{Q}}_p$ forms a subgroup $D_p$ that is called the decomposition group and is isomorphic to $\text{Gal}(\overline{\mathbb{Q}}_p/\mathbb{Q}_p)$. The definition of $D_p$ depends on the chosen embedding $\overline{\mathbb{Q}} \hookrightarrow \overline{\mathbb{Q}}_p$, or equivalently the chosen prime ideal $\overline{p}$, and so $D_p$ is only well-defined modulo conjugations induced by elements of $\text{Gal}(\overline{\mathbb{Q}}/\mathbb{Q})$. The quotient of $\overline{\mathbb{Z}}_p$ (the ring of integers of $\overline{\mathbb{Q}}_p$) by its maximal ideal is isomorphic to $\overline{\mathbb{F}}_p$ (the algebraic closure of $\mathbb{F}_p$) \cite{SerreLF}. Moreover we have a short exact sequence \cite{SerreLF}
\begin{equation}
\begin{tikzcd}
0 \arrow[r] & I_p \arrow[r] & D_p \arrow[r] & \text{Gal}(\overline{\mathbb{F}}_p/\mathbb{F}_p) \arrow[r] & 0
\end{tikzcd},
\end{equation}
where $I_p$ is called the inertia group. The Galois group $\text{Gal}(\overline{\mathbb{F}}_p/\mathbb{F}_p)$ is a profinite group
\begin{equation}
\text{Gal}(\overline{\mathbb{F}}_p/\mathbb{F}_p)=\varprojlim_{n}\text{Gal}(\mathbb{F}_{p^n}/\mathbb{F}_p),
\end{equation}
which is topologically generated by the Frobenius map 
\begin{equation} \label{eq:arithmeticFrobenius}
\fr_p:\overline{\mathbb{F}}_p \rightarrow \overline{\mathbb{F}}_p,~x \mapsto x^p.
\end{equation}
The map $\fr_p$ is also called the arithmetic Frobenius element, while its inverse $\fr_p^{-1}$ is called the geometric Frobenius element. 
\begin{remark}
For a discussion of the absolute Galois group $\text{Gal}(\overline{K}/K)$ of a number field $K$ and its decomposition groups, inertia groups, Frobenius elements, etc, the reader is referred to \cite{SerreLF}. 
\end{remark}

\subsection{Classical Weil Cohomology Theories} \label{sec:puremotiverealizations}

\noindent  We will now briefly discuss three classical Weil cohomology theories for smooth algebraic varieties: the Betti, de Rham, and \'etale cohomologies. First we must introduce the notion of a pure Hodge structure \cite{PetersSteenbrink}.

A pure Hodge structure $H$ with weight $l\in \mathbb{Z}$ consists of the following data:
\begin{enumerate}
	\item  A finite dimensional rational vector space $H_{\mathbb{Q}}$;
	
	\item  A decreasing filtration $F^*H$ of the complex vector space $H_{\mathbb{C}}:=H_{\mathbb{Q}} \otimes_{\mathbb{Q}} \mathbb{C}$,
\end{enumerate}
such that $H_{\mathbb{C}}$ admits a decomposition
\begin{equation}
H_{\mathbb{C}}=\oplus_{p+q=l}\, H^{p,q},
\end{equation}
where $H^{p,q}:=F^{p} \cap \overline{F}^q$ \cite{PetersSteenbrink}. Here complex conjugation is defined with respect to the real structure $H_{\mathbb{R}}:=H_{\mathbb{Q}} \otimes_{\mathbb{Q}} \mathbb{R}$ of $H_{\mathbb{C}}$. The definition immediately implies that 
\begin{equation}
F^k=\oplus_{p \geq k}\, H^{p,\,l-p}.
\end{equation}
The category of all pure Hodge structures will be denoted by $\textbf{HS}_\mathbb{Q}$; it is a semi-simple abelian category. The simplest example of a pure Hodge structure is the Hodge-Tate object $\mathbb{Q}(n),n \in \mathbb{Z}$ with weight $-2\,n$.
\begin{definition}
The finite dimensional vector space of the Hodge-Tate object $\mathbb{Q}(n)$ is
\begin{equation}
(2 \pi i)^n \mathbb{Q} \subset \mathbb{C}
\end{equation}
and its Hodge decomposition is
\begin{equation}
\mathbb{Q}(n)^{-n,-n}=(2 \pi i)^n \mathbb{Q} \otimes_{\mathbb{Q}} \mathbb{C}.
\end{equation}
\end{definition}

Now we can move on to Weil cohomology theories. Let $X$ be a smooth projective variety defined over a number field $K$. 
\begin{enumerate}
	\item \textbf{Betti cohomology.} Suppose $\sigma: K \hookrightarrow \mathbb{C}$ is an embedding of $K$ into $\mathbb{C}$, so that by extension of the base field $X$ defines a variety over $\mathbb{C}$, namely $X \times_{\sigma} \mathbb{C}$. The $\mathbb{C}$-valued points (classical points) of $X \times_{\sigma} \mathbb{C}$, denoted by $(X \times_\sigma \mathbb{C})(\mathbb{C})$, form a smooth projective complex manifold. Intuitively, $X$ is defined by polynomials with coefficients in the field $K$. Using the embedding $\sigma$, these coefficients are sent to $\mathbb{C}$, and the equations that define $X$ become polynomials with coefficients in $\mathbb{C}$, and hence define a smooth complex manifold. The Betti cohomology associated to the embedding $\sigma$ is the singular cohomology group
	\begin{equation}
	H^i_{B,\sigma}(X)(n):=H^i\big((X \times_\sigma \mathbb{C})(\mathbb{C}),~\mathbb{Q}(n) \big)=H^i\big((X \times_\sigma \mathbb{C})(\mathbb{C}),~\mathbb{Q} \big) \otimes \mathbb{Q}(n),
	\end{equation}
	where $\mathbb{Q}(n)$ is the Hodge-Tate object defined above \cite{PetersSteenbrink}. From Hodge theory, there exists a pure Hodge structure on $H_{B,\sigma}(X)(n)$ with weight $w=i-2\,n$, i.e. it has a Hodge decomposition
	\begin{equation}
	H^i_{B,\sigma}(X)(n) \otimes_{\mathbb{Q}} \mathbb{C}= \oplus_{p+q=w}\, H^{p,q},~~h^{p,q}:=\text{dim}_{\mathbb{C}}\,H^{p,q}.
	\end{equation}
	Moreover if $\sigma$ is real, i.e. the image of $\sigma$ is contained in $\mathbb{R}$, then complex conjugation $c \in \text{Gal}(\mathbb{C}/\mathbb{R})$ acts on the points of $(X \times_\sigma \mathbb{C})(\mathbb{C})$, and this action induces an involution $c^*$ of $H_{B,\sigma}(X)(n)$. Define $\phi_{\sigma}$ to be the involution on $H_{B,\sigma}(X)(n)$ induced by the action of $c$ on both the points $(X \times_\sigma \mathbb{C})(\mathbb{C})$ and the coefficient ring $\mathbb{Q}(n)$. Then the conjugate-linear involution $\phi_{\sigma} \otimes c$ preserves the Hodge decomposition of $H_{B,\sigma}(X)(n) \otimes \mathbb{C}$, i.e. it sends $H^{p,q}$ to $H^{p,q}$.
	
	\item \textbf{de Rham cohomology.} Over a variety $X$, there exists a complex of sheaves of algebraic differential forms \cite{Hartshorne}
	\begin{equation} 
	\Omega_{X/K}^*: 0 \rightarrow \mathcal{O}_{X/K} \xrightarrow{d} \Omega_{X/K}^1 \xrightarrow{d} \cdots \xrightarrow{d}  \Omega_{X/K}^{\text{dim}(X)} \rightarrow 0.
	\end{equation}
	In order to define a `reasonable' cohomology theory, we have to choose an injective resolution $\Omega_{X/K}^* \rightarrow I^* $ in the  category of complexes of sheaves on $X$. The hypercohomology of $\Omega_{X/K}^*$ is defined to be \cite{Voisin}
	\begin{equation}
	\mathbb{H}^i(X_{\text{Zar}},\Omega^*_{X/K})=H^i(\Gamma(X,I^*)),
	\end{equation}
	where $X_{\text{Zar}}$ means the Zariski topology on $X$. The de Rham cohomology of $X$ shifted by $n$ is the hypercohomology of the shifted complex $\Omega^*_{X/K}[n]$
	\begin{equation}
	H^i_{\text{dR}}(X)(n):=\mathbb{H}^i(X_{\text{Zar}},\Omega^*_{X/K}[n]),~\text{with}~(\Omega^*_{X/K}[n])^l=\Omega^{l+n}_{X/K},
	\end{equation}
	which is a finite dimensional $K$-vector space \cite{Voisin}. The de Rham cohomology $H_{\text{dR}}(X)(n)$ has a decreasing filtration $F^p H_{\text{dR}}(X)(n)$ given by
	\begin{equation}
	F^p H^i_{\text{dR}}(X)(n)=\mathbb{H}^i(X_{\text{Zar}}, F^p\Omega_{X/K}^*[n]),
	\end{equation}
	where the complex $F^p\Omega_{X/K}^*[n]$ is
	\begin{equation}
	F^p\, \Omega_{X/K}^*[n]: 0 \rightarrow \cdots \rightarrow 0 \rightarrow \Omega_{X/K}^{p+n} \xrightarrow{d}  \Omega_{X/K}^{p+1+n} \xrightarrow{d} \cdots \xrightarrow{d}  \Omega_{X/K}^{\text{dim}\,X} \rightarrow 0.
	\end{equation}
	
	\item $\ell$\textbf{-adic cohomology.} Let $\ell$ be a prime number. The $\ell$-adic cohomology of $X$ is defined by the inverse limit
	\begin{equation}
	H^i_{\text{\'et}}(X_{\overline{K}},\mathbb{Q}_\ell):=\varprojlim_n H^i((X \times_K \overline{K})_{\text{\'et}}, \mathbb{Z}/\ell^n \mathbb{Z}) \otimes_{\mathbb{Z}_\ell} \mathbb{Q}_{\ell},
	\end{equation}
	where $(X \times_K \overline{K})_{\text{\'et}}$ means the \'etale topology on the $\overline{K}$-variety $X_{\overline{K}}:=X \times_K \overline{K}$ and $\mathbb{Z}/\ell^n \mathbb{Z}$ means the constant \'etale torsion sheaf on $(X \times_K \overline{K})_{\text{\'et}}$. The $\ell$-adic cyclotomic character $\mathbb{Q}_\ell(1)$ is defined by the inverse limit
	\begin{equation}
	\mathbb{Q}_\ell(1):= \varprojlim_n \mu_{\ell^n}(\overline{K}) \otimes_{\mathbb{Z}_{\ell}} \mathbb{Q}_\ell,
	\end{equation}
	where $\mu_{\ell^n}(\overline{K})$ consists of the $\ell^n$-th roots of unity and admits an action by $\mathbb{Z}/\ell \mathbb{Z}$. Let $\mathbb{Q}_\ell(n)$ be the tensor product $\mathbb{Q}_\ell(1)^{\otimes n}$, which is a continuous representation of the absolute Galois group $\text{Gal}(\overline{K}/K)$ \cite{Taylor}. The $\ell$-adic cohomology $H^i_{\text{\'et}}(X_{\overline{K}},\mathbb{Q}_\ell)(n)$ is defined by
	\begin{equation}
	H^i_{\text{\'et}}(X_{\overline{K}},\mathbb{Q}_\ell)(n):=H^i_{\text{\'et}}(X_{\overline{K}},\mathbb{Q}_\ell) \otimes_{\mathbb{Q}_\ell} \mathbb{Q}_{\ell}(n),
	\end{equation}  
	which is a continuous representation of $\text{Gal}(\overline{K}/K)$  \cite{MilneEC}. 
\end{enumerate}

\noindent There are standard comparison isomorphisms between the three cohomology theories \cite{Nekovar}:
\begin{enumerate}
	\item There is an isomorphism 
	\begin{equation} \label{eq:bettiderhamComparison}
	I_\sigma:H^i_{B,\sigma}(X)(n) \otimes_{\mathbb{Q}} \mathbb{C} \rightarrow H^i_{\text{dR}}(X)(n) \otimes_{\sigma} \mathbb{C},
	\end{equation}
	between the Betti and de Rham cohomologies that preserves the Hodge filtration, i.e. it sends $\oplus_{k \geq p} H^{k,w-k}$ to $F^pH_{\text{dR}}(X)(n) \otimes_{\sigma} \mathbb{C}$. The isomorphism $I_\sigma$ clearly depends on the choice of $\sigma$. If further $\sigma$ is a real embedding, $I_\sigma$ sends the involution $\phi_{\sigma} \otimes c$ on the left hand side to the involution $1 \otimes c$ on the right hand. 
	
	\item Let the embedding $\overline{\sigma}: \overline{K} \hookrightarrow \mathbb{C}$ be an extension of $\sigma$. There is an isomorphism \begin{equation} \label{eq:ladiccomparison}
	I_{\ell,\overline{\sigma}}:H^i_{B,\sigma}(X)(n)  \otimes_{\mathbb{Q}} \mathbb{Q}_\ell \rightarrow H^i_{\text{\'et}}(X_{\overline{K}},\mathbb{Q}_\ell)(n),
	\end{equation}
	between the Betti and $\ell$-adic cohomologies. This isomorphism clearly depends on the choice of $\overline{\sigma}$. If $\sigma$ is in particular a real embedding, complex conjugation $c$ defines an element $\overline{\sigma}^*(c)$ in $\text{Gal}(\overline{K}/K)$. Then $I_{\ell,\overline{\sigma}}$ sends the involution $\phi_\sigma \otimes 1$ on the left hand side to the involution $\overline{\sigma}^*(c)$ on the right hand side.
\end{enumerate}

\noindent 
The two comparison isomorphisms imply
\begin{equation}
\text{dim}_{\mathbb{Q}}H^i_{B,\sigma}(X)(n)=\text{dim}_KH^i_{\text{dR}}(X)(n)=\text{dim}_{\mathbb{Q}_\ell}H^i_{\text{\'et}}(X_{\overline{K}},\mathbb{Q}_\ell)(n).
\end{equation}

\begin{example}
	The three classical realizations of the second cohomology of $\mathbb{P}^1_{K}$ are given by:
	\begin{enumerate}
		\item $H^2_{B,\sigma}(\mathbb{P}^1_{K})=\mathbb{Q}(-1)_{\text{B}}=(2 \pi i)^{-1}\, \mathbb{Q}$, which has a Hodge decomposition of type $(1,1)$.
		
		\item $H^2_{\text{dR}}(\mathbb{P}^1_{K})=\mathbb{Q}(-1)_{\text{dR}}=K$, with Hodge filtration given by $F^2=0$ and $F^{1}=K$.
		
		\item $H^2_{\text{\'et}}(\mathbb{P}_{\overline{K}},\mathbb{Q}_\ell)=\mathbb{Q}_\ell(-1)$.
	\end{enumerate}
\end{example}

\subsection{\texorpdfstring{$L$}{L}-functions}

\noindent We will now define the $L$-function of a smooth projective variety $X$ over a number field $K$. For convenience, we will denote the $\ell$-adic cohomology $H^i_{\text{\'et}}(X_{\overline{K}},\mathbb{Q}_\ell)$ by $\mathbf{M}_\ell$, which is a continuous representation of the absolute Galois group $\text{Gal}(\overline{K}/K)$. Recall that a non-archimedean prime $v$ of $\mathcal{O}_K$ is given by a nonzero prime ideal of $\mathcal{O}_K$ \cite{SerreLF, Neukirch}. Suppose $I_v$ is the inertia group of $v$ in $\text{Gal}(\overline{K}/K)$. We say $\mathbf{M}_\ell$ is unramified at $v$ if the action of $I_v$ on $\mathbf{M}_\ell$ is trivial, in which case the geometric Frobenius element has a well defined action on $\mathbf{M}_\ell$ that will be denoted by $\Phi_v$ \cite{SerreLF, Taylor}. Since $X$ is a smooth projective variety, $\mathbf{M}_{\ell}$ is pure of weight $i$. Here `pure' means that there exists a set $S$ consisting of finitely many primes such that for a nonarchimedean prime $v \notin S$ which does not divide $\ell$, the representation $\mathbf{M}_\ell$ is unramified at $v$ and all the eigenvalues of $\Phi_v$ are algebraic numbers with absolute values $\text{Nm}(v)^{w/2}$ \cite{DeligneWeil}, where $\text{Nm}$ is the norm map defined on the fractional ideals of $\mathcal{O}_{K}$ \cite{Neukirch, SerreLF}. Generally, for a nonarchimedean prime $v$ of $\mathcal{O}_K$ such that $\ell \nmid \text{Nm}(v)$, let $\mathbf{M}_\ell^{I_v}$ be the subspace of $\mathbf{M}_\ell$ that is invariant under the action of $I_v$. Then the geometric Frobenius element has a well-defined action on $\mathbf{M}_\ell^{I_v}$ that will also be denoted by $\Phi_v$. The characteristic polynomial of $\mathbf{M}$ at $v$ is defined by
\begin{equation}
P_v(\mathbf{M}_\ell,T)=\text{det}\big( 1-T\,\Phi_v \vert \mathbf{M}_\ell^{I_v} \big),~ \ell \nmid \text{Nm}(v).
\end{equation}
From Deligne's proof of the Weil conjectures \cite{DeligneWeil}, if $X$ has good reduction at the non-archimedean prime $v$, then:
\begin{enumerate}
	\item $P_v(\mathbf{M}_\ell,T)$ is an integral polynomial of $\mathbb{Z}[T]$ and it is independent of the choice of $\ell$.
	
	\item $P_v(\mathbf{M}_\ell,T)$ has a factorization of the form
	\begin{equation}
	P_v(\mathbf{M}_\ell,T)= \prod_{j=1}^{\text{dim}(\mathbf{M}_\ell)}(1-\alpha_j \, T),
	\end{equation}
	where $\alpha_j$ is an algebraic integer with $|\alpha_j|=\text{Nm}(v)^{w/2}$ for every $j$.
\end{enumerate}
The variety $X$ has bad reduction at only finitely many primes, and Serre has a conjecture about the behavior of $P_v(\mathbf{M}_\ell,T)$ at these bad primes \cite{SerreLF}:
\begin{conjecture} \label{eq:Serresconjecture}
	For an arbitrary non-archimedean prime $v$, $P_v(\mathbf{M}_\ell,T)$ lies in $\mathbb{Z}[T]$, and it does not depend on the choice of $\ell$. The integral polynomial $P_v(\mathbf{M}_\ell,T)$ has a factorization
	\begin{equation}
	P_v(\mathbf{M}_\ell,T)= \prod_{j=1}^{\text{dim}( \mathbf{M}_\ell^{I_v})}(1-\alpha_j \, T),
	\end{equation}
	where for every $j$, $\alpha_j$ is an algebraic integer with absolute value
	\begin{equation}
	|\alpha_j|=\text{Nm}(v)^{w_j/2},~0 \leq w_j \leq w.
	\end{equation}
	
\end{conjecture}

Assuming this conjecture, the local $L$-factor of $\mathbf{M}_\ell$ at $v$ is defined by
\begin{equation} \label{eq:defnlocalLfactor}
L_v(\mathbf{M}_\ell,s):=P^{-1}_v(\mathbf{M}_\ell,\text{Nm}(v)^{-s}),
\end{equation}
and the $L$-function of $\mathbf{M}_\ell$ is defined by the infinite product
\begin{equation} \label{eq:defnLfunction}
L(\mathbf{M}_\ell,s):= \prod_v L_v(\mathbf{M}_\ell,s),
\end{equation}
where the product is over all non-archimedean primes of $\mathcal{O}_K$. The local $L$-factor $L_v(\mathbf{M}_\ell,s)$ satisfies \cite{Nekovar}
\begin{equation} \label{eq:transportLfunction}
L_v(\mathbf{M}_\ell(m),s)=L_v(\mathbf{M}_\ell,m+s)\text{ and }L_v(\mathbf{M}_{1,\ell} \oplus \mathbf{M}_{2,\ell},s)=L_v(\mathbf{M}_{1,\ell},s)L_v(\mathbf{M}_{2,\ell},s),
\end{equation} 
and the $L$-function of $\mathbf{M}_\ell$ satisfies similar indentities. Deligne's theorem and Conjecture \eqref{eq:Serresconjecture} imply that the infinite product in the definition of $L(\mathbf{M}_\ell,s)$ \eqref{eq:defnLfunction} converges absolutely when $\text{Re}(s) > w/2+1$, and hence $L(\mathbf{M}_\ell,s)$ is a nowhere vanishing holomorphic function in this region.

\subsection{A Universal Cohomology Theory}

\noindent Suppose $X$ is a smooth algebraic variety defined over a field $k$. For simplicity, we will focus on the case where the characteristic of $k$ is zero. As discussed in Section \ref{sec:puremotiverealizations}, there are several cohomology theories of $X$, e.g. the Betti, de Rahm, and \'etale cohomologies of $X$, that are very ``similar" to each other under the comparison isomorphisms. The idea of pure motives, originally from Grothendieck, is to formalize the similarity between different Weil cohomology theories. More precisely, the goal is to construct an abelian category of pure motives whose incarnations yield the classical Weil cohomologies. We now explain the intuition underlying the construction of the category of pure motives. 

One philosophy in the study of category theory is that we should focus on the morphisms between objects. Let us formally denote the cohomology of $X$ by $h(X)$. Unlike the classical Weil cohomology theories, here $h(X)$ is an abstract object without inner structure, like an atom in chemistry. To construct the category of pure motives, we will need to construct the morphisms between two objects $h(X)$ and $h(Y)$, where $h(Y)$ is the cohomology of another smooth variety $Y$. If $f:Y \rightarrow X$ is a map from $Y$ to $X$, we expect there is a morphism $f^*: h(X) \rightarrow h(Y)$. Importantly, the graph of such an $f$ defines an algebraic cycle in $X \times Y$ with codimension equal to $\text{dim}\,X$. Therefore the morphisms between $h(X)$ and $h(Y)$ will be constructed from algebraic cycles of $X \times Y$ with codimension equal to $\text{dim}\,X$, which can be considered as multi-valued maps from $Y$ to $X$.

To recover something reminiscent of the classical Weil cohomology theories, we will want to break the cohomology of $X$ into finer sub-objects, i.e.
\begin{equation}
h(X)=\oplus_i h^i(X).
\end{equation} 
The idea is that the object $h(X)$ is determined by the identity morphism $\text{Id}:X \rightarrow X$, which corresponds to the diagonal $\Delta_X$ of $X \times X$. Breaking up the object $h(X)$ is equivalent to decomposing the diagonal $\Delta_X$ into the sum of suitable algebraic cycles of $X \times X$. To make this idea more precise, we need to briefly review the theory of algebraic cycles and adequate equivalence relations \cite{FultonI}. 

\subsection{Algebraic Cycles and Adequate Equivalence Relations}

\noindent Let $\textbf{SmProj}/k$ be the category of non-singular projective varieties over a field $k$. A prime cycle $Z$ of a non-singular projective variety $X$ is an irreducible algebraic subvariety, and its codimension is defined as $\text{dim}\,X-\text{dim}\,Z$. On the other hand, an irreducible closed subset of $X$ has a natural algebraic variety structure induced by that of $X$ \cite{Hartshorne}. The set of prime cycles of dimension $r$ (resp. codimension $r$) generates a free abelian group that will be denoted by $C_r(X)$ (resp. $C^r(X)$), and elements of $C_r(X)$ (resp. $C^r(X)$) will be called the algebraic cycles of dimension $r$ (resp. codimension $r$). Two prime cycles $Z_1$ and $Z_2$ are said to intersect with each other properly if
\begin{equation}
\text{codim}(Z_1 \cap Z_2) =\text{codim}(Z_1)+\text{codim}(Z_2),
\end{equation} 
where $Z_1 \cap Z_2$ means the set-theoretic intersection between $Z_1$ and $Z_2$. If two prime cycles $Z_1$ and $Z_2$ intersect with each other properly, the intersection product $Z_1 \cdot Z_2$ is defined as
\begin{equation}
Z_1 \cdot Z_2= \sum_T m(T;\,Z_1 \cdot Z_2)\, T,
\end{equation}
where the sum is over all irreducible components of $Z_1 \cap Z_2$ and $m(T;\,Z_1 \cdot Z_2)$ is Serre's intersection multiplicity formula \cite{FultonI}. Extending the definition by linearity, we have that the intersection product is defined for algebraic cycles $Z=\sum_j\,m_j\,Z_j$ and $W= \sum_l\, n_l\,W_l$ when $Z_j$ and $W_l$ intersect properly for all $j$ and $l$. Therefore there is a partially defined intersection product on algebraic cycles
\begin{equation}
\begin{tikzcd}
C^r(X) \times C^s(X) \arrow[r,dotted] &C^{r+s}(X).
\end{tikzcd}
\end{equation}

If $f:X \rightarrow Y$ is a morphism between two non-singular projective varieties $X$ and $Y$, the pushforward homomorphism $f_*$ on algebraic cycles is defined by
\begin{equation}
f_*(Z):=
\begin{cases}
0 &\text{ if}~ \text{dim}\, f(Z) < \text{dim}\, Z, \\
[k(Z):k(f(Z))] \cdot f(Z)        &~ \text{if}~ \text{dim}\, f(Z) = \text{dim}\, Z,
\end{cases}
\end{equation}
where $Z$ is a prime cycle and $k(Z)$ (resp. $k(f(Z))$) is the function field of $Z$ (resp. $f(Z)$) \cite{Hartshorne}. Here $ [k(Z):k(f(Z))] $ is the degree of the field extension. Now we want to define the pullback homomorphism $f^*$. Given a prime cycle $W$ of $Y$, the first attempt is to naively try 
\begin{equation} \label{eq:cyclepullback}
f^*(W):= \sum_{T \subset f^{-1}(W)} \ell_{\mathcal{O}_{X,T}}(\mathcal{O}_{f^{-1}(W),T}) \cdot T,
\end{equation}
where the sum is over the irreducible components of $f^{-1}(Z)$ and $\ell_{\mathcal{O}_{X,T}}(\mathcal{O}_{f^{-1}(Z),T})$ is the length of $\mathcal{O}_{f^{-1}(Z),T}$ in $\mathcal{O}_{X,T}$ \cite{FultonI}. However this definition is only partially defined and in general $f^*(W)$ does not make sense. The solution to the above problems is to find an equivalence relation $\sim$ on the algebraic cycles such that the quotient group $C^*(X)/\sim$ is well behaved.
\begin{definition}
An equivalence relation $\sim$ on the algebraic cycles is called an adequate equivalence relation if given two arbitrary cycles $Z_1$ and $Z_2$, there exists a cycle $Z'_1$ in the equivalence class of $Z_1$ such that $Z'_1$ intersects with $Z_2$ properly, while the equivalence class of the intersection $Z'_1 \cdot Z_2$ is independent of the choice of $Z'_1$.
\end{definition}

Hence for an adequate equivalence relation $\sim$, there is a well defined intersection product on the quotient group $C_{\sim}^*(X):=C^*(X)/\sim$
\begin{equation}
C^r(X)_{\sim} \times C^s(X)_{\sim}  \rightarrow  C^{r+s}(X)_{\sim}.
\end{equation}
After quotienting, the pushforward and pullback homomorphisms are also well defined \cite{FultonI}:
\begin{equation}
f_*:C_{r,\sim}(X) \rightarrow C_{r,\sim}(Y), ~~f^*:C^r_{\sim}(Y) \rightarrow C^r_{\sim}(X).
\end{equation}
The set of adequate equivalence relations is ordered in the way such that $\sim_1$ is said to be finer than $\sim_2$ if for every cycle $Z$, $Z \sim_1 0$ implies $Z \sim_2 0$. The most important adequate equivalence relations are rational and numerical equivalence, which are the finest and coarsest adequate equivalence relations, respectively \cite{Andre, Scholl, FultonI}. When we are explicitly taking $\sim$ to be one of these relations, we will write e.g. $C^r_{\text{num}}$. In all three classical Weil cohomology theories, there exists a cycle class map $\text{cl}$
\begin{equation}
\text{cl}:C^*_{\text{rat}}(X)_{\mathbb{Q}} \rightarrow H^*(X),
\end{equation}
which doubles the degree and sends the intersection product of cycles to the cup product of cohomology classes.

\subsection{Pure Motives}\label{sec:motives}

\noindent Although the three examples of classical Weil cohomology theories we discussed above behave as if they all arise from an algebraically defined cohomology theory over $\mathbb{Q}$, it is known that this is not correct \cite{Andre, Scholl}. Grothendieck's idea to explain this phenomenon is that there exists one cohomology theory that is universal in the sense that all Weil cohomology theories are realizations of it. More precisely, Grothendieck conjectured that there exists a rigid tensor abelian category $\textbf{M}_{\text{hom}}(k,\mathbb{Q})$ over $\mathbb{Q}$ and a functor $h$
\begin{equation}
h: \left( \textbf{SmProj}/k\right)^{\text{op}} \rightarrow \textbf{M}_{\text{hom}}(k,\mathbb{Q})
\end{equation} 
such that for every Weil cohomology theory $H^*$, there exists a functor $H^*_m$ that factors through $h$
\begin{equation}
\begin{tikzcd}
\left( \textbf{SmProj}/k \right)^{\text{op}} \arrow[rd,"H^*"'] \arrow[r, "h"] & \textbf{M}_{\text{hom}}(k,\mathbb{Q}) \arrow[d,dotted,"H^*_m"] \\
& \text{Gr}^{\geq 0}\, \text{Vec}_K 
\end{tikzcd}.
\end{equation}
Here `op' means the opposite category and $\text{Gr}^{\geq 0}\, \text{Vec}_K$ is the abelian category of graded vector spaces over the field $K$. Now we will introduce the construction of the category of motives $\textbf{M}_{\sim}(k,\mathbb{Q})$, where $\sim$ is rational or numerical equivalence \cite{Andre, Scholl}. Given two non-singular projective varieties $X$ and $Y$, the group of correspondences from $X$ to $Y$ with degree $r$ is defined as
\begin{equation}
\text{Corr}^r(X,Y):=C^{\text{dim}\,X+r}(X \times Y).
\end{equation}
The composition of correspondences 
\begin{equation}
\text{Corr}^r(X,Y) \times \text{Corr}^s(Y,Z) \rightarrow \text{Corr}^{r+s}(X,Z)
\end{equation}
is defined by  
\begin{equation}
g \times h \rightarrow h \circ g:=(p_{13})_*\big((p_{12})^*g\cdot (p_{23})^*h\big),
\end{equation}
where $p_{12}$ is the natural projection morphism from $X \times Y \times Z$ to $X \times Y$, etc \cite{Scholl}. For a morphism $f:Y \rightarrow X$, its graph $\Gamma_f$ in $X \times Y$ is an algebraic variety that is isomorphic to $Y$, so that $\Gamma_f$ is an element of $\text{Corr}^0(X,Y)$ \cite{Hartshorne}. A correspondence of $\text{Corr}^0(X,Y)$ can be seen as a multi-valued morphism from $Y$ to $X$. A correspondence $\gamma$ defines a homomorphism from $H^*(X)$ to $H^*(Y)$ by
\begin{equation} \label{eq:cyclemap}
\gamma_*: x \mapsto p_{2,*}\,(\,p_1^*\, x \, \cup \, \text{cl}(\gamma)),
\end{equation}
where $p_1$ (resp. $p_2$) is the projection morphism from $X \times Y$ to $X$ (resp. $Y$). The homomorphism $(\Gamma_f)_*$ induced by $\Gamma_f$ is just the pullback homomorphism $f^*$.

We can now explicitly describe the construction of the category $\textbf{M}_{\sim}(k,\mathbb{Q})$, which proceeds in\ three steps \cite{Andre, Scholl}:

\begin{enumerate}
	\item First we construct a category whose objects are formal symbols
	\begin{equation}
	\{h(X):X \in \textbf{SmProj}/k \}.
	\end{equation}
	The morphisms between two objects are given by
	\begin{equation}
	\text{Hom}(h(X),h(Y)):=\text{Corr}^0_{\sim}(X,Y)_{\mathbb{Q}},
	\end{equation}  
	where we have defined 
	\begin{equation}
	\text{Corr}^r_{\sim}(X,Y)=\text{Corr}^0(X,Y)/\sim\text{ and }\text{Corr}^0_{\sim}(X,Y)_{\mathbb{Q}}=\text{Corr}^0_{\sim}(X,Y) \otimes_{\mathbb{Z}} \mathbb{Q}.
	\end{equation}
	This category can be seen as the linearization of $\left(\textbf{SmProj}/k\right)^{\text{op}}$. For example, for the projective line $\mathbb{P}^1$, we have 
	\begin{equation}
	\text{Corr}^0_{\sim}(\mathbb{P}^1,\mathbb{P}^1)_{\mathbb{Q}}=\mathbb{Q} \left(\mathbb{P}^1 \times \{ 0\} \right)+\mathbb{Q} \left( \{ 0\} \times \mathbb{P}^1 \right),
	\end{equation}
	i.e. it is the 2 dimensional vector space spanned by the horizontal line $e_2:=\mathbb{P}^1 \times \{ 0\} $ and the vertical line $ e_0:=\{ 0\} \times \mathbb{P}^1$. Notice that $e_0$ corresponds to the constant map from $\mathbb{P}^1$ to itself.
	\item Take the pseudo-abelianization of the category constructed in Step 1 and denote this new category by $\textbf{M}^{\text{eff}}_{\sim}(k,\mathbb{Q})$. More explicitly, the objects of $\textbf{M}^{\text{eff}}_{\sim}(k,\mathbb{Q})$ are formally
	\begin{equation}
	\{(h(X),e):X \in \textbf{SmProj}/k\,\text{and}\,\, e \in \text{Corr}^0_{\sim}(X,X)_{\mathbb{Q}} , \,e^2=e\},
	\end{equation}
	 and the morphisms between two objects are given by
	\begin{equation}
	\text{Hom}((h(X),e),(h(Y),f)):=f \circ \text{Corr}^0_{\sim}(X,Y)_{\mathbb{Q}} \circ e .
	\end{equation}
	The category constructed in Step 1 admits a natural embedding into $\textbf{M}^{\text{eff}}_{\sim}(k,\mathbb{Q})$, which comes from sending the object $h(X)$ to $(h(X), \Delta_X)$. Here $\Delta_X$ is the graph of the identity map of $X$, i.e. the diagonal of $X \times X$. The graph of the identity morphism of $\mathbb{P}^1$, i.e. $\Delta_{\mathbb{P}^1}$, is rationally equivalent to $e_0+e_2$. If we define 
    \begin{equation}
    h^0(\mathbb{P}^1):=(h(\mathbb{P}^1),e_0),~h^2(\mathbb{P}^1):=(h(\mathbb{P}^1),e_2),
    \end{equation}
    the object $(h(\mathbb{P}^1),\Delta_{\mathbb{P}^1})$ has a decomposition given by \cite{Andre, Scholl}
	\begin{equation}
	(h(\mathbb{P}^1),\Delta_{\mathbb{P}^1})=h^0(\mathbb{P}^1) \oplus h^2(\mathbb{P}^1).
	\end{equation}
	The component $h^0(\mathbb{P}^1)$ is also denoted by $\mathbb{Q}(0)$, while the component $h^2(\mathbb{P}^1)$ is also denoted by $\mathbb{Q}(-1)$. This is the motivic version of the property of the classical cohomology groups of $\mathbb{P}^1$. $(h(\mathbb{P}^1),e_0)$ is the zero-th cohomology because $e_0$ corresponds to the constant map of $\mathbb{P}^1$, which only preserves the zero-th cohomology while killing all higher degree cohomologies.
	
	\item  The category $\textbf{M}_{\sim}(k,\mathbb{Q})$ is constructed from $\textbf{M}^{\text{eff}}_{\sim}(k,\mathbb{Q})$ by inverting the object $\mathbb{Q}(-1)$. The objects of $\textbf{M}_{\sim}(k,\mathbb{Q})$ are formally
	\begin{equation}
	\{(h(X),e,m):X \in \textbf{SmProj}/k,\,\, e \in \text{Corr}^0_{\sim}(X,X)_{\mathbb{Q}} , \,e^2=e,\,\,\text{and}\,\,m \in \mathbb{Z}\},
	\end{equation}
	and the morphisms between two objects are given by
	\begin{equation}
	\text{Hom}((h(X),e,m),(h(Y),f,n)):=f \circ \text{Corr}^{n-m}_{\sim}(X,Y)_{\mathbb{Q}} \circ e.
	\end{equation}
	The category $\textbf{M}^{\text{eff}}_{\sim}(k,\mathbb{Q})$ is isomorphic to the full subcategory of $\textbf{M}_{\sim}(k,\mathbb{Q})$ generated by objects of the form $(h(X),e,0)$.
\end{enumerate}

Given two objects of $\textbf{M}_{\sim}(k,\mathbb{Q})$, the morphisms between them form a rational vector space, which is finite dimensional if $\sim$ is numerical equivalence. We can define a direct sum and direct product in $\textbf{M}_{\sim}(k,\mathbb{Q})$ by \cite{Scholl} \subeqs{(h(X),e,m) \oplus (h(Y),f,m)&:=(h(X \amalg Y),e \oplus f,m) \\ (h(X),e,m) \otimes (h(Y),f,n)&:=(h(X \times Y),e \times f,m+n),} respectively. We can also define the dual of an element of $\textbf{M}_{\sim}(k,\mathbb{Q})$ by
\begin{equation}
(h(X),e,m)^{\vee}:=(h(X),e^\top,\text{dim} \, X-m)
\end{equation}
where $e^\top$ means the transpose of $e$. In fact, the object $\mathbb{Q}(0)$ is a unit of $\textbf{M}_{\sim}(k,\mathbb{Q})$ \cite{Andre, Scholl}. The dual of $\mathbb{Q}(-1)$ is denoted by $\mathbb{Q}(1)$, and the Tate motive $\mathbb{Q}(m)$ is defined by
\begin{equation}
\mathbb{Q}(m):=
\begin{cases}
\mathbb{Q}(1)^m &\text{ if}~ m \geq 0, \\
\mathbb{Q}(-1)^{-m}        &~ \text{if}~m <0.
\end{cases}
\end{equation}
The Tate twist of a pure motive $M \in \textbf{M}_{\sim}(k,\mathbb{Q})$ by $\mathbb{Q}(m)$ is defined by
\begin{equation}
M(m):=M \otimes \mathbb{Q}(m).
\end{equation}
We have a further operation on the category $\textbf{M}_{\sim}(k,\mathbb{Q})$ called the extension of field. If the field $k'$ is an arbitrary extension of $k$, then $X \times_k k'$ is a non-singular variety defined over $k'$. After an extension of the base field from $k$ to $k'$ in the construction of $\textbf{M}_{\sim}(k,\mathbb{Q})$, we obtain a functor  
\begin{equation}
\textbf{M}_{\sim}(k,\mathbb{Q}) \times_k k' \rightarrow \textbf{M}_{\sim}(k',\mathbb{Q}).
\end{equation}
From the construction of $\textbf{M}_{\sim}(k,\mathbb{Q})$, there is a functor
\begin{equation}
h:\left( \textbf{SmProj}/k \right)^{\text{op}} \rightarrow \textbf{M}_{\sim}(k,\mathbb{Q}),
\end{equation}
which sends $X$ to $(h(X),\Delta_X,0)$ and $f: Y \rightarrow X$ to $\Gamma_f$.

\subsection{The Hodge Conjecture and Motivic Splits} \label{sec:hodgeconjecturereview}

\noindent We will now formulate the Hodge conjecture in the language of pure motives. From the construction of pure motives in Section \ref{sec:motives}, we deduce that every Weil cohomology theory $H^*$ automatically factors through $\textbf{M}_{\text{rat}}(k,\mathbb{Q})$:
\begin{equation}
\begin{tikzcd}
\left( \textbf{SmProj}/k \right)^{\text{op}} \arrow[rd,"H^*"'] \arrow[r, "h"] & \textbf{M}_{\text{rat}}(k,\mathbb{Q}) \arrow[d,dotted,"H^*_{\text{rat}}"] \\
& \text{Gr}^{\geq 0}\, \text{Vec}_K 
\end{tikzcd}.
\end{equation}
However, the category $\textbf{M}_{\text{rat}}(k,\mathbb{Q})$ is not abelian \cite{Andre, Scholl}. On the other hand, $\textbf{M}_{\text{num}}(k,\mathbb{Q})$ is known to be abelian and semi-simple \cite{Jannsen,Andre, Scholl}, but it is not known whether an arbitrary Weil cohomology theory $H^*$ will factor through $\textbf{M}_{\text{num}}(k,\mathbb{Q})$. If an algebraic cycle $\gamma$ is numerically equivalent to 0, then it will define a zero morphism in $\textbf{M}_{\text{num}}(k,\mathbb{Q})$. Thus, in order for $H^*$ to factor through $\textbf{M}_{\text{num}}(k,\mathbb{Q})$, we need the induced homomorphism $\gamma_*$ in the formula \eqref{eq:cyclemap} to be zero. This statement is historically known as Grothendieck's \textbf{Conjecture D}, and has not yet been proven \cite{Kleiman}.

\noindent \textbf{Conjecture D} \textit{If an algebraic cycle $\gamma$ is numerically equivalent to 0, then $\text{cl}(\gamma)$ is zero for every Weil cohomology theory.}

\noindent This conjecture also implies that the homological equivalence relation $\sim_{H^*}$ defined by a Weil cohomology theory $H^*$ is the same as numerical equivalence.

Assuming that \textbf{Conjecture D} holds, we have a functor 
\begin{equation} \label{eq:bettirealization}
H^*_{\text{B},\sigma}: \textbf{M}_{\text{num}}(k,\mathbb{Q}) \rightarrow \text{Gr}^{\geq 0}\, \text{Vec}_\mathbb{Q}
\end{equation}
In fact, since the Betti cohomology $H^i_{\text{B},\sigma}(X)$ is endowed with a pure Hodge structure, the functor in the formula \eqref{eq:bettirealization} is lifted to the functor
\begin{equation} \label{eq:hodgerealization}
\mathfrak{R}_{\sigma}: \textbf{M}_{\text{num}}(k,\mathbb{Q}) \rightarrow \textbf{HS}_{\mathbb{Q}},
\end{equation}
(with $\textbf{HS}_{\mathbb Q}$ as defined in \S1.2), which is called the Hodge realization functor \cite{Andre}. The Hodge conjecture can be succinctly stated as follows.

\noindent  \textit{For every algebraically closed field $F$ which admits an embedding $\sigma_F:F \hookrightarrow \mathbb{C}$, the Hodge realization functor
\begin{equation}
\mathfrak{R}_{\sigma_F}: \textbf{M}_{\text{num}}(F,\mathbb{Q}) \rightarrow \textbf{HS}_{\mathbb{Q}},
\end{equation}
is full-faithful.}

\noindent Assuming \textbf{Conjecture D} and the Hodge conjecture, we have an easy corollary that is crucial to this paper.
\begin{corollary} \label{splittingcorollary}
Suppose $k$ is a number field with an embedding $\sigma:k \rightarrow \mathbb{C}$ and $\mathbf{M}$ is a motive in $ \textbf{M}_{\text{num}}(k,\mathbb{Q}) $. If the Hodge realization of $\mathbf{M}$ splits into
\begin{equation}
\mathfrak{R}_{\sigma}(\mathbf{M})=H' \oplus H''
\end{equation}
in the category $ \textbf{HS}_{\mathbb{Q}}$, then there exists a field $k_1$, which is a finite extension of $k$, such that over $k_1$ we have
\begin{equation}
\mathbf{M}=\mathbf{M}' \oplus \mathbf{M}''~\text{in}~ \textbf{M}_{\text{num}}(k_1,\mathbb{Q}) ,~\mathfrak{R}_{\sigma_1}(\mathbf{M}')=H',~ \mathfrak{R}_{\sigma_1}(\mathbf{M}'')= H''.
\end{equation}
Here $\sigma_1:k_1 \hookrightarrow \mathbb{C}$ is an extension of $\sigma:k \hookrightarrow \mathbb{C}$.
\end{corollary}
\begin{proof}
The split of $H=H' \oplus H''$ is given by an idempotent homomorphism $e_R:H \rightarrow H$. Let $\overline{\sigma}:\overline{k} \rightarrow \mathbb{C}$ be an extension of $\sigma$. From the Hodge Conjecture, the Hodge realization functor
\begin{equation}
\mathfrak{R}_{\overline{\sigma}}: \textbf{M}_{\text{num}}(\overline{k},\mathbb{Q}) \rightarrow \textbf{HS}_{\mathbb{Q}}
\end{equation}
is full-faithful, so that there exists an idempotent morphism $e:\mathbf{M} \rightarrow \mathbf{M}$ in the category $\textbf{M}_{\text{num}}(\overline{k},\mathbb{Q}) $. Since $\textbf{M}_{\text{num}}(\overline{k},\mathbb{Q})$ is an abelian category, the idempotent morphism $e$ yields a split
\begin{equation} \label{eq:splittingofM}
\mathbf{M}=\mathbf{M}' \oplus \mathbf{M}'',
\end{equation}
and from the construction of this split, we have 
\begin{equation}
\mathfrak{R}_{\overline{\sigma}}(\mathbf{M}')=H',~ \mathfrak{R}_{\overline{\sigma}}(\mathbf{M}'')= H''.
\end{equation}
This morphism $e$ is defined by an algebraic cycle over $\overline{k}$, and hence it is defined over a field $k_1$ that is a finite extension of $k$. Thus the split in the formula \eqref{eq:splittingofM} happens already in the category $ \textbf{M}_{\text{num}}(k_1,\mathbb{Q}) $, without needing to take the algebraic closure of $k$.
\end{proof}

We can now explain the relationship between the split of the $\ell$-adic cohomology of $X$ in  \eqref{eq:fluxEtaleSplit}, and the Hodge conjecture. If we assume Grothendieck's \textbf{Conjecture D} and the Hodge conjecture, then Corollary \ref{splittingcorollary} implies the existence of a number field $K_1$, a finite extension of $K$, such that in the category $\textbf{M}_{\text{num}}(K_1,\mathbb{Q})$ of pure motives over $K_1$, the pure motive $h^3(X)\times_K K_1$ splits into the direct sum
\begin{equation} \label{eq:splitofpuremotive}
h^3(X)\times_K K_1=\mathbf{M}_{\text{flux}} \oplus \mathbf{M}_{\text{remainder}},
\end{equation}
where the Hodge realization of $\mathbf{M}_{\text{flux}}$ (resp. $\textbf{M}_{\text{remainder}}$) is $\textbf{H}_{\text{flux}}$ (resp. $\textbf{H}_{\text{remainder}}$). The $\ell$-adic realization of this split immediately yields the split of the $\ell$-adic cohomology of $X$ in \eqref{eq:fluxEtaleSplit}. However, this argument depends critically on Grothendieck's \textbf{Conjecture D} and Hodge conjecture, which are two deep and unproven conjectures.

%% file: v2.bbl
\begin{thebibliography}{10}

\bibitem{Langlands}
R.~P. Langlands.
\newblock Problems in the theory of automorphic forms.
\newblock In C.~T. Taam, editor, {\em Lectures in Modern Analysis and
  Applications III}, pages 18--61, Berlin, Heidelberg, 1970. Springer Berlin
  Heidelberg.

\bibitem{langlandsBook}
J.~Bernstein, S.S. Kudla, S.~Gelbart, E.~Kowalski, E.~de~Shalit, D.~Gaitsgory,
  J.W. Cogdell, and D.~Bump.
\newblock {\em {An Introduction to the Langlands Program}}.
\newblock Birkh{\"a}user Boston, 2013.

\bibitem{Diamond}
Fred Diamond.
\newblock {On Deformation Rings and Hecke Rings}.
\newblock {\em Annals of Mathematics}, 144(1):137--166, 1996.

\bibitem{Conrad}
Brian Conrad, Fred Diamond, and Richard Taylor.
\newblock {Modularity of Certain Potentially Barsotti-Tate Galois
  Representations}.
\newblock {\em Journal of the American Mathematical Society}, 12(2):521--567,
  1999.

\bibitem{Taylor2}
Richard Taylor and Andrew Wiles.
\newblock {Ring-Theoretic Properties of Certain Hecke Algebras}.
\newblock {\em Annals of Mathematics}, 141(3):553--572, 1995.

\bibitem{Wiles}
Andrew Wiles.
\newblock {Modular Elliptic Curves and Fermat's Last Theorem}.
\newblock {\em Annals of Mathematics}, 141(3):443--551, 1995.

\bibitem{Breuil}
Christophe Breuil, Brian Conrad, Fred Diamond, and Richard Taylor.
\newblock {On the Modularity of Elliptic Curves over Q: Wild 3-Adic Exercises}.
\newblock {\em Journal of the American Mathematical Society}, 14(4):843--939,
  2001.

\bibitem{Schutt}
Matthias Schuett.
\newblock {Fields of definition of singular K3 surfaces}, 2006, math/0612396.

\bibitem{livne}
Ron Livn{\'e}.
\newblock Motivic orthogonal two-dimensional representations of
  gal($\bar{\mathbb{q}}â\mathbb{Q}$).
\newblock {\em Israel Journal of Mathematics}, 92(1):149--156, Feb 1995.

\bibitem{yui:rigid}
Fernando~Q. {Gouvea} and Noriko {Yui}.
\newblock {Rigid Calabi-Yau Threefolds over Q Are Modular}.
\newblock {\em arXiv e-prints}, page arXiv:0902.1466, Feb 2009, 0902.1466.

\bibitem{rev1}
Mariana Grana.
\newblock {Flux compactifications in string theory: A Comprehensive review}.
\newblock {\em Phys. Rept.}, 423:91--158, 2006, hep-th/0509003.

\bibitem{rev2}
Michael~R. Douglas and Shamit Kachru.
\newblock {Flux compactification}.
\newblock {\em Rev. Mod. Phys.}, 79:733--796, 2007, hep-th/0610102.

\bibitem{rev3}
Frederik Denef, Michael~R. Douglas, and Shamit Kachru.
\newblock {Physics of String Flux Compactifications}.
\newblock {\em Ann. Rev. Nucl. Part. Sci.}, 57:119--144, 2007, hep-th/0701050.

\bibitem{rev4}
Frederik Denef.
\newblock {Les Houches Lectures on Constructing String Vacua}.
\newblock {\em Les Houches}, 87:483--610, 2008, 0803.1194.

\bibitem{candelas:attractors}
Philip Candelas, Xenia de~la Ossa, Mohamed Elmi, and Duco van Straten.
\newblock A one parameter family of calabi-yau manifolds with attractor points
  of rank two, 2019, 1912.06146.

\bibitem{dewolfe:11222}
Oliver DeWolfe.
\newblock {Enhanced symmetries in multiparameter flux vacua}.
\newblock {\em JHEP}, 10:066, 2005, hep-th/0506245.

\bibitem{kadir:octic}
Shabnam~Nargis Kadir.
\newblock {\em {The Arithmetic of Calabi-Yau manifolds and mirror symmetry}}.
\newblock PhD thesis, Oxford U., 2004, hep-th/0409202.

\bibitem{ribet:review}
Kenneth~A Ribet.
\newblock {Galois representations and modular forms}.
\newblock {\em Bulletin of the American Mathematical Society}, 32(4):375--402,
  1995.

\bibitem{meyer:book}
C.~Meyer.
\newblock {\em {Modular Calabi-Yau Threefolds}}.
\newblock Fields Institute monographs. American Mathematical Soc.

\bibitem{yui:review}
Noriko Yui.
\newblock {Modularity of Calabi--Yau varieties: 2011 and beyond}, 2012,
  1212.4308.

\bibitem{123}
K.~Ranestad, J.H. Bruinier, G.~van~der Geer, G.~Harder, and D.~Zagier.
\newblock {\em {The 1-2-3 of Modular Forms: Lectures at a Summer School in
  Nordfjordeid, Norway}}.
\newblock Universitext. Springer Berlin Heidelberg, 2008.

\bibitem{interactiveTable}
{LMFDB - The L-functions and Modular Forms Database}.
\newblock http://www.lmfdb.org.

\bibitem{Dwork}
Bernard Dwork.
\newblock On the rationality of the zeta function of an algebraic variety.
\newblock {\em American Journal of Mathematics}, 82(3):631--648, 1960.

\bibitem{DeligneWeil}
Pierre Deligne.
\newblock {La conjecture de Weil : I}.
\newblock {\em Publications Math\'ematiques de l'IH\'ES}, 43:273--307, 1974.

\bibitem{MilneEC}
J.S. Milne.
\newblock {\em {{\'E}tale Cohomology (PMS-33)}}.
\newblock Number v. 33 in Princeton Mathematical Series. Princeton University
  Press, 2016.

\bibitem{candelas:finite2}
Philip Candelas, Xenia de~la Ossa, and Fernando Rodriguez~Villegas.
\newblock {Calabi-Yau manifolds over finite fields. 2.}
\newblock {\em Fields Inst. Commun.}, 38:121--157, 2013, hep-th/0402133.

\bibitem{kachru:gkp}
Steven~B. Giddings, Shamit Kachru, and Joseph Polchinski.
\newblock {Hierarchies from Fluxes in String Compactifications}.
\newblock {\em Physical Review D}, 66(10), Nov 2002.

\bibitem{kachru:0411}
Oliver DeWolfe, Alexander Giryavets, Shamit Kachru, and Washington Taylor.
\newblock {Enumerating flux vacua with enhanced symmetries}.
\newblock {\em JHEP}, 02:037, 2005, hep-th/0411061.

\bibitem{becker:m8fold}
Katrin Becker and Melanie Becker.
\newblock {M theory on eight manifolds}.
\newblock {\em Nucl. Phys.}, B477:155--167, 1996, hep-th/9605053.

\bibitem{dasgupta:mGflux}
Keshav Dasgupta, Govindan Rajesh, and Savdeep Sethi.
\newblock {M theory, orientifolds and G - flux}.
\newblock {\em JHEP}, 08:023, 1999, hep-th/9908088.

\bibitem{sen:fourfold}
Ashoke Sen.
\newblock {Orientifold limit of F theory vacua}.
\newblock {\em Phys. Rev.}, D55:R7345--R7349, 1997, hep-th/9702165.

\bibitem{gukov:superpotential}
Sergei Gukov, Cumrun Vafa, and Edward Witten.
\newblock {CFT's from Calabi-Yau four folds}.
\newblock {\em Nucl. Phys.}, B584:69--108, 2000, hep-th/9906070.
\newblock [Erratum: Nucl. Phys.B608,477(2001)].

\bibitem{moore:a&along}
Gregory~W. Moore.
\newblock {Arithmetic and Attractors}.
\newblock 1998, hep-th/9807087.

\bibitem{moore:arithmeticLectures}
Gregory~W. Moore.
\newblock {Strings and Arithmetic}.
\newblock In {\em {Proceedings, Les Houches School of Physics: Frontiers in
  Number Theory, Physics and Geometry II: On Conformal Field Theories, Discrete
  Groups and Renormalization: Les Houches, France, March 9-21, 2003}}, pages
  303--359, 2007, hep-th/0401049.

\bibitem{consani2000geometry}
Caterina Consani and Jasper Scholten.
\newblock {Geometry and arithmetic on a quintic threefold}.
\newblock 2000, math/0009134.

\bibitem{schutt:hilbert}
Luis {Dieulefait}, Ariel {Pacetti}, and Matthias {Schuett}.
\newblock {Modularity of the Consani-Scholten quintic}.
\newblock {\em arXiv e-prints}, page arXiv:1005.4523, May 2010, 1005.4523.

\bibitem{straten:hilbert}
Slawomir {Cynk}, Matthias {Sch{\"u}tt}, and Duco {van Straten}.
\newblock {Hilbert modularity of some double octic Calabi--Yau threefolds}.
\newblock {\em arXiv e-prints}, page arXiv:1810.04495, Oct 2018, 1810.04495.

\bibitem{taylor:CMfields}
Patrick~B. {Allen}, Frank {Calegari}, Ana {Caraiani}, Toby {Gee}, David {Helm},
  Bao~V. {Le Hung}, James {Newton}, Peter {Scholze}, Richard {Taylor}, and
  Jack~A. {Thorne}.
\newblock {Potential Automorphy over CM Fields}.
\newblock {\em arXiv e-prints}, page arXiv:1812.09999, Dec 2018, 1812.09999.

\bibitem{densityTheorem}
P.~Stevenhagen and H.~W. Lenstra.
\newblock Chebotar{\"e}v and his density theorem.
\newblock {\em The Mathematical Intelligencer}, 18(2):26--37, 1996.

\bibitem{candelas:2param1}
Philip Candelas, Xenia De~La~Ossa, Anamaria Font, Sheldon~H. Katz, and David~R.
  Morrison.
\newblock {Mirror symmetry for two parameter models. 1.}
\newblock {\em Nucl. Phys.}, B416:481--538, 1994, hep-th/9308083.
\newblock [AMS/IP Stud. Adv. Math.1,483(1996)].

\bibitem{meyer:table2}
C.~Meyer.
\newblock {Newforms of weight two for $\Gamma_0(N)$ with rational
  coefficients}.
\newblock http://meyer-idstein.de/weight2.pdf.

\bibitem{rolf:genus3}
Monika Lynker and Rolf Schimmrigk.
\newblock {Geometric Kac-Moody modularity}.
\newblock {\em J. Geom. Phys.}, 56:843--863, 2006, hep-th/0410189.

\bibitem{verrill:ruled}
Klaus Hulek and Helena Verrill.
\newblock {On the modularity of Calabi-Yau threefolds containing elliptic ruled
  surfaces}, 2005, math/0502158.

\bibitem{kachru:fluxThreefolds}
Alexander Giryavets, Shamit Kachru, Prasanta~K Tripathy, and Sandip~P Trivedi.
\newblock {Flux Compactifications on Calabi-Yau Threefolds}.
\newblock {\em Journal of High Energy Physics}, 2004(04):003Ð003, Apr 2004.

\bibitem{lee:quinticModular}
Edward {Lee}.
\newblock {A modular quintic Calabi-Yau threefold of level 55}.
\newblock {\em arXiv e-prints}, page arXiv:0903.1140, Mar 2009, 0903.1140.

\bibitem{verrill:a4}
Klaus {Hulek} and Helena {Verrill}.
\newblock {On modularity of rigid and nonrigid Calabi-Yau varieties associated
  to the root lattice A\_4}.
\newblock {\em arXiv Mathematics e-prints}, page math/0304169, Apr 2003,
  math/0304169.

\bibitem{kondo:worldsheetLFunctions}
Satoshi Kondo and Taizan Watari.
\newblock {String-theory Realization of Modular Forms for Elliptic Curves with
  Complex Multiplication}.
\newblock {\em Commun. Math. Phys.}, 367(1):89--126, 2019, 1801.07464.

\bibitem{kondo:worldsheetLFunctions2}
Satoshi Kondo and Taizan Watari.
\newblock Modular parametrization as polyakov path integral: Cases with cm
  elliptic curves as target spaces, 2019, 1912.13294.

\bibitem{Kadir:2010dh}
Shabnam Kadir, Monika Lynker, and Rolf Schimmrigk.
\newblock {String Modular Phases in Calabi-Yau Families}.
\newblock {\em J. Geom. Phys.}, 61:2453--2469, 2011, 1012.5807.

\bibitem{candelas:finite1}
Philip Candelas, Xenia de~la Ossa, and Fernando Rodriguez-Villegas.
\newblock {Calabi-Yau manifolds over finite fields. 1.}
\newblock 2000, hep-th/0012233.

\bibitem{candelas:dwork}
Philip Candelas and Xenia de~la Ossa.
\newblock {The Zeta-Function of a p-Adic Manifold, Dwork Theory for
  Physicists}.
\newblock {\em Commun. Num. Theor. Phys.}, 1:479--512, 2007, 0705.2056.

\bibitem{GreenePlesser}
Brian~R. Greene and M.~R. Plesser.
\newblock {Duality in {Calabi-Yau} Moduli Space}.
\newblock {\em Nucl. Phys.}, B338:15--37, 1990.

\bibitem{Neukirch}
J.~Neukirch and N.~Schappacher.
\newblock {\em {Algebraic Number Theory}}.
\newblock Grundlehren der mathematischen Wissenschaften. Springer Berlin
  Heidelberg, 2013.

\bibitem{SerreLF}
J.P. Serre.
\newblock {\em {Local Fields}}.
\newblock Graduate Texts in Mathematics. Springer New York, 2013.

\bibitem{Taylor}
Richard Taylor.
\newblock {Galois Representations}, 2002, math/0212403.

\bibitem{PetersSteenbrink}
C.A.M. Peters and J.H.M. Steenbrink.
\newblock {\em {Mixed Hodge Structures}}.
\newblock Ergebnisse der Mathematik und ihrer Grenzgebiete. 3. Folge / A Series
  of Modern Surveys in Mathematics. Springer Berlin Heidelberg, 2008.

\bibitem{Hartshorne}
R.~Hartshorne.
\newblock {\em {Algebraic Geometry}}.
\newblock Graduate Texts in Mathematics. Springer New York, 2013.

\bibitem{Voisin}
C.~Voisin and L.~Schneps.
\newblock {\em {Hodge Theory and Complex Algebraic Geometry I}}.
\newblock Number v. 1 in Cambridge Studies in Advanced Mathematics. Cambridge
  University Press, 2002.

\bibitem{Nekovar}
Jan Nekov{\'a}r.
\newblock {BeilinsonÕs conjectures}.
\newblock {\em Motives}, 55(1):537--570, 1994.

\bibitem{FultonI}
William Fulton.
\newblock {\em {Intersection Theory}}, volume~2.
\newblock Springer Science \& Business Media, 2013.

\bibitem{Andre}
Yves Andr{\'e}.
\newblock {Une Introduction Aux Motifs: Motifs Purs, Motifs Mixtes,
  P{\'e}riodes}.
\newblock {\em Panoramas et synth{\`e}ses-Soci{\'e}t{\'e} math{\'e}matique de
  France}, (17):1--258, 2004.

\bibitem{Scholl}
Anthony~J Scholl.
\newblock {Classical Motives}.
\newblock In {\em Proc. Symp. Pure Math}, volume~55, pages 163--187, 1994.

\bibitem{Jannsen}
Uwe Jannsen.
\newblock {Motives, Numerical Equivalence, and Semi-Simplicity}.
\newblock {\em Inventiones mathematicae}, 107(1):447--452, 1992.

\bibitem{Kleiman}
Steven~L Kleiman.
\newblock {The Standard Conjectures}.
\newblock {\em Motives (Seattle, WA, 1991)}, 55:3--20, 1994.

\end{thebibliography}
